\documentclass[journal,twoside,web]{ieeecolor}
\usepackage{generic}
\usepackage{cite}
\usepackage{amsmath,amssymb,amsfonts}
\usepackage{algorithm}
\usepackage[noend]{algpseudocode}
\usepackage{graphicx}
\usepackage{textcomp}
\def\BibTeX{{\rm B\kern-.05em{\sc i\kern-.025em b}\kern-.08em
    T\kern-.1667em\lower.7ex\hbox{E}\kern-.125emX}}
\newtheorem{lemma}{Lemma}
\newtheorem{definition}{Definition}
\newtheorem{theorem}{Theorem}
\newtheorem{problem}{Problem}
\newtheorem{proposition}{Proposition}
\newtheorem{corollary}{Corollary}
\newtheorem{assumption}{Assumption}

\begin{document}
\title{A Semi-Algebraic Framework for Verification and Synthesis of Control Barrier Functions}
\author{Andrew Clark \IEEEmembership{Senior Member, IEEE}
\thanks{A. Clark is with the Department of Electrical and Systems Engineering, Washington University in St. Louis,  St. Louis, MO 63130 USA (e-mail: andrewclark@wustl.edu). }
}

\maketitle

\begin{abstract}
Safety is a critical property for control systems in medicine, transportation, manufacturing, and other applications, and can be defined as ensuring positive invariance of a predefined safe set. This paper investigates the problems of verifying positive invariance of a semi-algebraic set as well as synthesizing sets that can be made positive invariant through Control Barrier Function (CBF)-based control. The key to our approach consists of mapping conditions for positive invariance to sum-of-squares constraints via the Positivstellensatz from real algebraic geometry. Based on these conditions, we propose a framework for verifying safety of CBF-based control including single CBFs, high-order CBFs, multi-CBFs, and systems with trigonometric dynamics and actuation constraints.
 In the area of synthesis, we propose algorithms for constructing CBFs, namely, an alternating-descent approach and  a local CBF approach. We evaluate our approach through  case studies on quadrotor UAV and power converter test systems.
\end{abstract}


\section{Introduction}
\label{sec:intro}
Safety is a critical property of autonomous control systems in medical, manufacturing, transportation, and other vital applications. From a control-theoretic perspective, safety is typically characterized by positive invariance of a particular safe region, for example, ensuring that a vehicle remains a desired safe distance from humans or obstacles, or maintaining the temperature of a system within tolerated limits. Methodologies such as energy-based methods \cite{dawson2023safe}, counterexample-guided synthesis \cite{henzinger2003counterexample}, and abstraction-based verification \cite{ratschan2007safety} have been proposed to verify system safety and design safe control algorithms. %When ensuring safety of a controlled system, two natural questions arise. First, Given a subset of the state space, does there exist a control policy that renders the set positive invariant (\emph{safety verification})? And second, how to construct a subset of the state space that can be rendered positive invariant through control (\emph{safety synthesis})? 

Conditions for positive invariance of controlled systems have been studied for decades beginning with the seminal work of Nagumo \cite{mejstrik2012some}. Exact conditions for positive invariance of a set have been proposed based on tangential conditions, i.e., whether the control input causes the state trajectory to evolve in a direction that is tangent to the boundary of the safe region \cite{aubin2011viability}. At present, however, there is currently a lack of \emph{computational} techniques for verifying that these invariance conditions are satisfied. Verification methodologies  have been proposed based on special cases, for example, when the safe set is the super-level set of a single polynomial and the control policy is given \cite{prajna2007framework}. These methods, however, do not readily generalize to verifying more complicated safe sets (e.g., intersections and unions of super-level sets) as well as incorporating actuation constraints.

In this paper, we develop a framework for verification of safety constraints that are expressed as invariance of semi-algebraic sets. A semi-algebraic set is any set that can be expressed as a combination of polynomial equalities and inequalities. The key challenge is that safety verification relies on checking \emph{non-existence} of certain conditions, namely, that there do not exist points on the boundary of the safe region that violate the tangential conditions. To overcome this challenge, we develop our framework based on theorems of alternatives, which map non-existence conditions of a primal problem to equivalent existence conditions of a dual problem. In particular, we leverage the Positivstellensatz from real algebraic geometry \cite{bochnak2013real} and Farkas lemma from linear algebra \cite{dinh2014farkas} to formulate a collection of sum-of-squares (SOS) programs. The solutions to these programs comprise certificates that validate safety of a given system. Furthermore, we use a novel construction based on semi-algebraic triangulations to construct continuous feedback controllers that satisfy the safety criteria, and extend our approach to systems with trigonometric functions in their dynamics.

While our proposed approach is agnostic to the control law that is used to ensure safety, in practice, our techniques are  motivated by the recent research interest in Control Barrier Functions (CBFs) for safe control \cite{ames2019control}. CBF-based controllers attempt to ensure safety by applying  linear constraints on the control input at each time step that enforce the tangent cone conditions for positive invariance, however, safety may still be violated if a control input satisfying the constraints does not exist. Verification of a semi-algebraic set can be interpreted as verification of a CBF-based control policy, in which different CBFs are used in different regions of the state space. We make this connection explicit in Section \ref{subsec:feedback-control}, where we derive conditions for existence of CBF-based control laws for a given safe region, and in Section \ref{subsec:verification-applications}, where we show how existing CBF constructions can be verified within our framework.

We also consider the problem of synthesizing controlled invariant sets. We propose two heursitic algorithms. The first algorithm is an alternating-descent approach based on the SOS programs that we developed for verification. The second algorithm constructs an invariant set by linearizing around an equilibrium point. %We validate our synthesis and verification methodologies through simulation study on a linearized quadrotor model with actuation constraints.

The paper is organized as follows. Section \ref{sec:related} presents the related work. Section \ref{sec:model} presents the system model and background material. Section \ref{sec:verification} presents a convex framework for verification of semi-algebraic sets. Section \ref{sec:synthesis} presents algorithms for synthesizing invariant sets. Section \ref{sec:simulation} presents simulation results. Section \ref{sec:conclusion} concludes the paper.
\section{Related Work}
\label{sec:related}
%Control Barrier Functions (CBFs) were first proposed in \cite{ames2014control,ames2016control} and have been widely studied as a framework for verifiably safe control. CBFs have been proposed for high-degree \cite{xiao2021high},  actuation-limited \cite{breeden2021high}, Euler-Lagrange \cite{cortez2020correct} time-delayed \cite{singletary2020control}, and uncertain \cite{xiao2020adaptive} systems.
Safety and positive invariance of nonlinear systems has been studied extensively, dating back to the seminal work of Nagumo \cite{mejstrik2012some}. Much of the work in this domain focuses on conditions for the state of a dynamical system to remain within a given set and existence of (possibly discontinuous) controllers that guarantee invariance (viability) \cite{aubin1991viability}. For detailed treatments of these bodies of work see \cite{blanchini2008set,aubin2011viability}. This paper builds on these conditions and focuses on algorithmic approaches to verifying viability and invariance of nonlinear systems. Specifically, our main contribution is to consider the sub-class of semi-algebraic sets and polynomial systems and give exact, convex conditions for viability. 

While in principle our framework is agnostic to the type of control law used to ensure invariance of the safe set, we are especially motivated by the extensive recent work on Control Barrier Functions (CBFs). Control Barrier Functions (CBFs) were first proposed in the seminal works \cite{ames2014control,ames2016control} and have been widely studied as a framework for verifiably safe control. CBFs have been proposed for high-degree \cite{xiao2021high},  actuation-limited \cite{breeden2021high}, Euler-Lagrange \cite{cortez2020correct}, time-delayed \cite{singletary2020control}, and uncertain \cite{xiao2020adaptive} systems. As we discuss in the paper, once a semi-algebraic set has been synthesized and/or verified as viable using our approach, CBF-based control laws can be used to ensure safety. We note that our main focus is on selecting and verifying safe invariant sets, making other related works on selecting parameters for CBFs \cite{parwana2022trust,gao2023learning} and analyzing performance and stability of CBF-based control laws \cite{reis2020control,mestres2022optimization} orthogonal to our paper.

Sum-of-squares optimization has been used to verify properties of uncontrolled dynamical systems, starting from construction and verification of Lyapunov functions \cite{papachristodoulou2002construction} and safety verification through barrier certificates \cite{prajna2007framework}. Recently, sum-of-squares optimization has been used to verify safety of control systems, for example, verification of CBF-based control policies \cite{ames2019control}. Most of the existing works in this area assume that a single CBF or HOCBF is used \cite{jagtap2020compositional}, although some recent efforts can incorporate multiple barrier functions \cite{schneeberger2023sos}. Importantly, these existing efforts verify safety by first synthesizing a nominal control policy, and then verifying that the control policy satisfies the conditions of the CBF. As a result, the safety guarantees are highly dependent on the choice of the nominal control policy. In contrast, this paper presents exact conditions for viability of a semi-algebraic set that do not depend on the nominal controller construction, and are general enough to include intersections and unions of invariant sets as well as trigonometric functions. Finally, we note that techniques other than sum-of-squares have been proposed for safety verification, such as sampling-based methods \cite{lavaei2021formal} or sufficient conditions \cite{xiao2022sufficient}.

%Verification of invariance via Psatz: original papers that didn't consider control. Verification with CBFs, often just a single CBF (ours more general). Multiple CBFs (cite Dorfler paper). Assume nominal controller.

In practical settings, safety guarantees may rely on the existence of continuous controllers that satisfy the desired safety properties. There are two general approaches to proving existence of such controllers. In the first class of approach, inspired by the original work of Artstein \cite{artstein1983stabilization}, a stabilizing controller is constructed without a closed-form representation. In the second class, a closed-form representation of the controller is provided \cite{sontag1989universal,ong2019universal}.  We consider the first class of existence proofs, and propose a proof of existence of continuous feedback control policies that guarantee invariance of simple algebraic sets. As a corollary, we give sufficient conditions for existence of CBF-based control policies.

%Existence of continuous control policies. Two classes: Artstein, or constructive (cite Cortes and Ames). Ours is the second.

A variety of methodologies have been proposed for synthesizing CBFs. Bilinear, alternating-descent sum-of-squares algorithms were proposed in \cite{dai2023convex}. Deep neural networks have shown significant promise for representing CBFs, due to their universality and the availability of efficient training algorithms \cite{dawson2023safe}. When the neural networks employ semi-algebraic activation functions, such as Rectified Linear Unit (ReLU), the resulting safe region is semi-algebraic and can be verified by our proposed approach. Related data-driven methods were proposed in \cite{salamati2024data}. In \cite{fisac2015reach}, it was shown that, when the value function is defined appropriately to capture safety and reachability properties of the system, barrier functions can be constructed by solving a Hamilton-Jacobi-Bellman equation. HJB-based algorithms for constructing value functions include discretizing the state space \cite{tonkens2022refining} and developing quadratic approximations via differential dynamic programming \cite{kumar2023fast}. In Section \ref{subsec:verification-applications}, we explain how our framework can be used to verify such methods.

The preliminary conference version of this paper presented an algebraic-geometric framework for verifying CBFs and  HOCBFs, but did not consider unions of CBFs, trigonometric functions, or actuation constraints~\cite{clark2021verification}. 

\section{Model and Preliminaries}
\label{sec:model}
This section presents the system model and definitions of safety. We also present background from real algebraic geometry.

\subsection{System Model and Safety Definitions}
\label{subsec:model}
We consider a nonlinear control system with dynamics 
\begin{equation}
\label{eq:dynamics}
\dot{x}(t) = f(x(t)) + g(x(t))u(t)
\end{equation}
where $x(t) \in \mathbb{R}^{n}$ denotes the state, $u(t) \in \mathcal{U} \subseteq \mathbb{R}^{m}$ is a control input, and $f: \mathbb{R}^{n} \rightarrow \mathbb{R}^{n}$ and $g: \mathbb{R}^{n} \rightarrow \mathbb{R}^{n \times m}$ are continuous functions. We assume that $\mathcal{U} = \{u : Au \leq c\}$ for some $A \in \mathbb{R}^{p \times m}$ and $c \in \mathbb{R}^{p}$. Throughout the paper, we let $[x]_{i}$ denote the $i$-th element of $x \in \mathbb{R}^{n}$. 

%The system is required to remain in a safe region $\mathcal{C}$, which is defined by $\mathcal{C} = \bigcap_{i=1}^{M}{\{x : h_{i}(x) \geq 0\}}$ for some functions $h_{1},\ldots,h_{M}: \mathbb{R}^{n} \rightarrow \mathbb{R}$. %We say that a control policy $\mu: \mathbb{R}^{n} \rightarrow \mathbb{R}^{m}$ ensures safety if there exists a set $\tilde{\mathcal{C}} \subseteq \mathcal{C}$ such that 
%We define a safe control policy as follows.
%\begin{definition}
%\label{def:safe-control-policy}
%A function $\mu: \mathbb{R}^{n} \rightarrow \mathbb{R}^{m}$ is a safe control policy if (i) there exists a set $\tilde{\mathcal{C}}$ with $\tilde{\mathcal{C}} \subseteq \mathcal{C}$, and (ii) for any $x(0) \in \tilde{\mathcal{C}}$, if $u(t) = \mu(x(t))$ for all $t \geq 0$, then $\tilde{\mathcal{C}}$ is positive invariant.
%\end{definition}

%A set is controlled positive invariant (CPI) if it is positive invariant for some control policy $\mu$.

%\subsection{Safety Definitions}
%\label{subsec:CBF-background}
In what follows, we define safety of the system (\ref{eq:dynamics}) by introducing the properties of viability and controlled positive invariance. We first consider an uncontrolled system $\dot{x}(t) = f(x(t))$.

\begin{definition}
    \label{def:PI}
    A set $\mathcal{D}$ is positive invariant for an uncontrolled system $\dot{x}(t) = f(x(t))$ if $x(0) \in \mathcal{D}$ implies that $x(t) \in \mathcal{D}$ for all time $t \geq 0$ when the solution to (\ref{eq:dynamics}) exists.
\end{definition}

Positive invariance implies that, if a system initially lies in set $\mathcal{D}$, it will remain in $\mathcal{D}$ for all future time. We have the following preliminary result before introducing conditions for positive invariance.

\begin{definition}[\cite{blanchini2008set}, Def. 4.6]
    \label{def:tangent-cone}
    For any set $\mathcal{D} \subseteq \mathbb{R}^{n}$, the tangent cone to $\mathcal{D}$ at $x \in \mathcal{D}$ is defined by $$\mathcal{T}_{\mathcal{D}}(x) = \left\{z : \lim\inf_{\tau \rightarrow 0}{\frac{\mbox{dist}(x + \tau z, \mathcal{D})}{\tau}} = 0\right\}.$$
\end{definition}

We note that, if $x$ is in the interior of $\mathcal{D}$, then $\mathcal{T}_{\mathcal{D}}(x) = \mathbb{R}^{n}$. We have the following result on the tangent cone for a class of sets $\mathcal{D}$.

\begin{lemma}[\cite{blanchini2008set}, Eq. (4.6)]
\label{lemma:tangent-cone-intersection}
Suppose that the set $\mathcal{D} = \{x : b_{i}(x) \geq 0, i=1,\ldots,M\}$ for some functions $b_{1},\ldots,b_{M}: \mathbb{R}^{n} \rightarrow \mathbb{R}$. Furthermore, suppose that, for all $x \in \mathcal{D}$, there exists $z$ such that $b_{i}(x) + \nabla b_{i}(x)z > 0$ for all $i=1,\ldots,M$. Then for any $x \in \partial \mathcal{D}$, $$\mathcal{T}_{\mathcal{D}}(x) = \left\{z: \frac{\partial b_{i}}{\partial x}z \geq 0 \ \forall i \mbox{ with } b_{i}(x) = 0\right\}.$$ 
%that satisfy (i) for all $x \in S$, there exists $z$ such that $b_{i}(x) + \nabla b_{i}(x)^{T}z < 0$ for all $i=1,\ldots,m$ and (ii) there exists a Lipschitz continuous vector field $\phi(x)$ such that $\nabla g_{i}(x)^{T}\phi(x) < 0$ for all $x \in \partial \mathcal{D}$. Then for any $x \in \partial \mathcal{D}$, $$\mathcal{T}_{\mathcal{D}}(x) = \{z: \frac{\partial b_{i}}{\partial x}z \geq 0 \ \forall i \mbox{ with } b_{i}(x) = 0\}.$$
\end{lemma}

We note that, if $\mathcal{D} = \{x : b(x) \geq 0\}$, then the condition of the lemma holds iff there does not exist $x$ with $b(x) = 0$ and $\frac{\partial b}{\partial x} = 0$. 
%When $\mathcal{D} = \{x: b(x) \geq 0\}$ for some differentiable function $b$, $\mathcal{T}_{\mathcal{D}}(x) = \{z : \frac{\partial b}{\partial x}z \geq 0 \}$ for all $x \in \partial \mathcal{D}$. A more general preliminary result on the tangent cone is as follows.
We now state Nagumo's Theorem, which gives exact conditions for positive invariance.

\begin{theorem}[Nagumo's Thoerem \cite{aubin2011viability}, Theorem 11.2.3]
\label{theorem:Nagumo}
Suppose that $\mathcal{D}$ is locally compact and $f$ is continuous on $\mathcal{D}$. Then the set $\mathcal{D}$ is positive invariant under dynamics $\dot{x}(t) = f(x(t))$ if and only if $f(x) \in \mathcal{T}_{\mathcal{D}}(x)$ for all $x \in \mathcal{D}$.
\end{theorem}

%In the case where $\mathcal{D} = \{x : b(x) \geq 0\}$, Theorem \ref{theorem:Nagumo} reduces to $\frac{\partial b}{\partial x}f(x) \geq 0$ for all $x$ with $b(x) = 0$. 

When the system is controlled, we define the viability property as follows.

\begin{definition}[\cite{aubin2011viability}, Ch. 11.1]
    \label{def:viability}
    A locally compact set $\mathcal{D}$ is \emph{viable} if, for all $x \in \partial \mathcal{D}$, there exists $u \in \mathcal{U}$ such that $(f(x) + g(x)u) \in \mathcal{T}_{\mathcal{D}}(x)$.
\end{definition}

Viability implies that, when the state reaches the boundary of the region $\mathcal{D}$, there exists a control input to ensure that the state remains in $\mathcal{D}$. Viability is a necessary condition for a set to be rendered positive invariant through control. However, it may also be desirable to ensure safety using a control algorithm with a particular structure, for example, a state feedback controller. We define  feedback controlled positive invariance to capture this property.

\begin{definition}
    \label{def:feedback-CPI}
    A locally compact set $\mathcal{D}$ is \emph{feedback controlled positive invariant} if there exists a continuous function $\mu: \mathcal{D} \rightarrow \mathcal{U}$ such that $(f(x) + g(x)\mu(x)) \in \mathcal{T}_{\mathcal{D}}(x)$ for all $x \in \partial \mathcal{D}$. 
\end{definition}

Clearly, feedback controlled positive invariance implies viability. Lastly, we define Control Barrier Functions (CBFs) as follows. Recall that a continuous function $\kappa: \mathbb{R} \rightarrow \mathbb{R}$ is class-K if $\kappa(0) = 0$  and $\kappa$ is strictly increasing.

%Control Barrier Functions (CBFs) have been proposed to ensure safety of nonlinear systems, and are defined as follows.

\begin{definition}
\label{def:CBF}
A function $b$ is a control barrier function for (\ref{eq:dynamics}) if $\frac{\partial b}{\partial x} \neq 0$ for all $x$ with $b(x) = 0$ and there is a class-K function $\kappa$ such that, for all $x$ with $b(x) \geq 0$, there exists $u$ satisfying 
\begin{equation}
\label{eq:CBF-def}
\frac{\partial b}{\partial x}(f(x) + g(x)u) \geq -\kappa(b(x)).
\end{equation}
\end{definition}
The following result establishes positive invariance for systems with CBFs.
\begin{theorem}[\cite{ames2019control}, Theorem 2] 
\label{theorem:CBF-safety}
Suppose that $b$ is a CBF, $b(x(0)) \geq 0$, and $u(t)$ satisfies (\ref{eq:CBF-def}) for all $t$. Then the set $\{x : b(x) \geq 0\}$ is positive invariant.
\end{theorem}

 \subsection{Preliminary Results}
 \label{subsec:preliminary}
 We now give  needed background results on real algebraic geometry. The results below can be found in \cite{bochnak2013real}. We let $\mathbb{R}[x]$ denote the set of polynomials over variable $x$ with real coefficients. 
 \begin{definition}[\cite{bochnak2013real}, Def. 2.1.4]
     \label{def:semialgebraic}
 A set $\mathcal{C} \subseteq \mathbb{R}^{n}$ is \emph{semialgebraic} if there exist integers $s$, $r_{1},\ldots,r_{s}$ and polynomials $b_{ij}(x) : \mathbb{R}^{n} \rightarrow \mathbb{R}$ for $i=1,\ldots,s, j=1,\ldots,r_{i}$ such that 
 \begin{equation}
 \label{eq:semialgebraic-def}
 \mathcal{C} = \bigcup_{i=1}^{s}{\bigcap_{j=1}^{r_{i}}{\{x : b_{ij}(x) \ast_{ij} 0\}}}
 \end{equation}
 where $\ast_{ij} \in \{=,>\}$.
  \end{definition}
The following result describes properties of closed semialgebraic sets.
\begin{lemma}
\label{lemma:closed-semialgebraic}
Any closed semialgebraic set $\mathcal{C}$ can be written in the form 
\begin{equation}
\label{eq:closed-semialgebraic}
\mathcal{C} = \bigcup_{i=1}^{s}{\bigcap_{j=1}^{r_{i}}{\{x : b_{ij}(x) \geq  0\}}}
\end{equation}
for some $s,r_{1},\ldots,r_{s} \in \mathbb{Z}_{>0}$ and polynomials $b_{ij}:\mathbb{R}^{n} \rightarrow \mathbb{R}$ for $i=1,\ldots,r$, $j=1,\ldots,r_{s}$.
\end{lemma}

If $s=1$, then $\mathcal{C}$ is a \emph{simple} semialgebraic set. A Sum-of-Squares (SOS) polynomial is a polynomial $f(x)$ such that $$f(x) = \sum_{i=1}^{k}{g_{i}(x)^{2}}$$ for some polynomials $g_{1}(x),\ldots,g_{k}(x)$. We use the notation $f \in SOS$ to mean that $f$ is a sum-of-squares polynomial.  Selecting coefficients of $f(x)$ to ensure that $f(x)$ is SOS can be represented as a semidefinite program, a procedure known as SOS optimization \cite{parrilo2003semidefinite,ahmadi2017improving}. A special case of SOS optimization is the SOS feasibility problem, which is stated as follows. Let $M$ and $N$ be positive integers and let $\eta_{ij},i=1,\ldots,N,j=1,\ldots,M$ and $\phi_{i},i=1,\ldots,N$ be polynomials in $\mathbb{R}^{n}$. The SOS feasibility problem consists of finding polynomials $\alpha_{1}(x),\ldots,\alpha_{M}(x)$ such that $$\left(\sum_{j=1}^{M}{\alpha_{j}(x)\eta_{ij}(x)} + \phi_{i}(x)\right) \in SOS$$
 for all $i=1,\ldots,N$. 
 
 The cone  $\Gamma[q_{1},\ldots,q_{k}]$ associated with polynomials $q_{1},\ldots,q_{k}$ is equal to the set of polynomials $f$ with $$f(x) = p_{0}(x) + \sum_{i=1}^{N}{p_{i}(x)\beta_{i}(x)},$$ where $p_{0},\ldots,p_{N}$ are SOS and $\beta_{1},\ldots,\beta_{N}$ are products of powers of the $q_{i}$'s. For a set of polynomials $f_{1},\ldots,f_{N} \in \mathbb{R}[x]$, we define the ideal generated by the polynomials as $$\mathcal{I}[f_{1},\ldots,f_{N}] = \left\{\sum_{i=1}^{N}{\eta_{i}f_{i}} : \eta_{1},\ldots,\eta_{N} \in \mathbb{R}[x]\right\}.$$ Finally, for a set of polynomials $\mathcal{H}=\{h_{1},\ldots,h_{s}\}$, the monoid generated by $\mathcal{H}$ is given by the products of powers of the $h_{i}$'s, i.e., $$\mathcal{M}[\mathcal{H}] = \left\{\prod_{i=1}^{s}{h_{i}(x)^{r_{i}}} : r_{1},\ldots,r_{s} \in \mathbb{Z}_{\geq 0}\right\}$$ 
 
 The Positivstellensatz, stated as follows, gives equivalent conditions for existence of solutions to systems of polynomial equations and inequalities. These equivalent conditions will be used to formulate equivalent SOS programs for the safety and invariance conditions to be defined in the following sections. 
 \begin{theorem}[Positivstellensatz \cite{bochnak2013real}, Proposition 4.4.1]
 \label{theorem:Psatz}
 Consider a collection of polynomials $\mathcal{F} = \{f_{j}: j=1,\ldots,r\}$, $\mathcal{G} = \{g_{i} : i=1,\ldots,m\}$ and $\mathcal{H} = \{h_{k} : k=1,\ldots,s\}$. Then the set 
 \begin{multline*}
 \left(\bigcap_{j=1}^{r}{\{x : f_{j}(x) \geq 0\}}\right) \cap \left(\bigcap_{i=1}^{m}{\{x : g_{i}(x) = 0\}}\right) \\
 \cap \left(\bigcap_{k=1}^{s}{\{x : h_{k}(x) \neq 0\}}\right)
 \end{multline*}
 is empty if and only if there exist polynomials $f$, $g$, and $h$ satisfying (i) $f(x) \in \Gamma[f_{1},\ldots,f_{r}]$, (ii) $g(x) \in \mathcal{I}[g_{1},\ldots,g_{m}]$,  (iii) $h \in \mathcal{M}[h_{1},\ldots,h_{s}]$, and (iv) $f(x) + g(x) + h(x)^{2} = 0$ for all $x$.
 \end{theorem}

When the $f_{j}$'s, $g_{i}$'s, $h_{k}$'s, and powers $a_{1},\ldots,a_{s}$ of the $h_{k}$'s are given, attempting to find solutions of the equation $f(x) + g(x) + h(x)^{2} = 0$ is an SOS feasibility problem \cite{parrilo2003semidefinite}. This problem can be written as 
\begin{multline*}
    \mbox{\textbf{(Psatz-SOS)}} \quad \left(\sum_{i=1}^{m}{\eta_{i}(x)g_{i}(x)} - \sum_{S \subseteq \{1,\ldots,r\}}{\alpha_{S}(x)\prod_{i \in S}{f_{i}(x)}} \right. \\
   \left. - \prod_{k=1}^{s}{h_{k}(x)^{2a_{k}}}\right) \in SOS \\
    \alpha_{S} \in SOS \ \forall S \subseteq \{1,\ldots,r\}
\end{multline*}
We let $\mbox{Psatz-SOS}(\mathcal{F}, \mathcal{G}, \mathcal{H})$ denote an instance of the above SOS feasibility problem. Finally, Farkas's Lemma is as follows.
 \begin{lemma}[Farkas's Lemma \cite{dinh2014farkas}]
 \label{lemma:Farkas}
 Let $A \in \mathbb{R}^{m \times n}$ and $b \in \mathbb{R}^{m}$. Then exactly one of the following is true:
 \begin{enumerate}
 \item There exists $x \in \mathbb{R}^{n}$ with $Ax \leq b$
 \item There exists $y \in \mathbb{R}^{m}$ with $A^{T}y = 0$, $b^{T}y < 0$, and $y \geq 0$.
 \end{enumerate}
 \end{lemma}
Farkas's Lemma gives equivalent conditions for existence of solutions of linear inequalities. 

\section{Problem Formulation: Safety Verification}
\label{sec:verification}
This section presents our approach for safety verification of semi-algebraic sets. We first consider the problem of verifying viability, and then the problem of verifying feedback controlled positive invariance. We discuss several special cases and implications for construction of CBF-based control policies.

As a preliminary, we define the concept of a \emph{practical} semi-algebraic set.

\begin{definition}
\label{def:practical-semialgebraic}
A closed semi-algebraic set $\mathcal{C}$ is practical if $\mathcal{C} = \bigcup_{i=1}^{r}{\mathcal{C}_{i}}$ where each $\mathcal{C}_{i}$ is a closed and simple semi-algebraic set with $$\mathcal{C}_{i} = \bigcap_{j=1}^{r_{i}}{\{x : b_{ij}(x) \geq 0\}}$$ and for all $x \in \mathcal{C}_{i}$, there exists $z \in \mathbb{R}^{n}$ satisfying $\frac{\partial b_{ij}}{\partial x}z + b_{ij}(x) > 0$ for all $j=1,\ldots,r_{i}$.
\end{definition}

In other words, $\mathcal{C}$ is a practical semi-algebraic set if it is a union of closed simple semi-algebraic sets, each of which satisfies the conditions of Lemma \ref{lemma:tangent-cone-intersection}. Techniques for verifying practicality of a semi-algebraic set are discussed in Appendix \ref{appendix:practical}.

\subsection{Viability Verification}
\label{subsec:viability-verification}
The problem studied in this section is formulated as follows.

\begin{problem}
    \label{problem:viability-verification}
    Given a practical semialgebraic set $\mathcal{C}$, verify whether $\mathcal{C}$ is viable under dynamics (\ref{eq:dynamics}) when $f$ and $g$ are polynomials in $x$.
\end{problem}

%In the remainder of this section, we assume that, for all $i=1,\ldots,s$, the functions $\{b_{ij}: j=1,\ldots,r_{i}\}$ satisfy conditions (i) and (ii) of Lemma \ref{lemma:tangent-cone-intersection}.

Our solution to Problem \ref{problem:viability-verification} is as follows. First, we characterize the boundary of a closed semialgebraic set, and then describe the tangent cone to each point in the boundary. Next, using Nagumo's Theorem, we develop necessary and sufficient conditions for viability of the set. Finally, we use the Positivstellensatz to construct an SOS program for verifying viability. All proofs can be found in Appendix \ref{appendix:proofs}.

\begin{theorem}
    \label{prop:semialgebraic-boundary}
    Let $\mathcal{C}$ be a practical semialgebraic set defined by polynomials $b_{ij}(x)$ as in (\ref{eq:closed-semialgebraic}). If $x \in \partial \mathcal{C}$, then there exists (i) $S \subseteq \{1,\ldots,s\}$, (ii) a set of integers $\{j_{i} : i \in \{1,\ldots,r\} \setminus S\}$, where $j_i \in \{1,\ldots,r_{i}\}$, and (iii) a collection of nonempty subsets $T_{i} \subseteq \{1,\ldots,r_{i}\}$ for $i \in S$ such that $b_{ij_{i}}(x) < 0$ for $i \notin S$, $b_{ij}(x) = 0$ for $(i,j) \in S \times T_{i}$, and $b_{ij}(x) > 0$ for $i \in S$, $j \notin T_{i}$.
\end{theorem}

In Theorem \ref{prop:semialgebraic-boundary}, the set $S$ corresponds to the set of $i$ such that $x \in \mathcal{C}_{i}$. For each $i \notin S$, $j_{i}$ is an index such that $b_{ij_{i}}(x) < 0$. For each $i \in S$, the set $T_{i}$ corresponds to the set of $j$ such that $b_{ij}(x) = 0$, while $\{1,\ldots,r_{i}\} \setminus T_{i}$ corresponds to the set of $j$ such that $b_{ij}(x) > 0$, implying that $x$ is in the interior of $\{b_{ij}(x) \geq 0\}$.

\textbf{Example 1:} As an example of the notation in Theorem \ref{prop:semialgebraic-boundary}, let $\mathcal{C} = \mathcal{C}_{1} \cup \mathcal{C}_{2} \subseteq \mathbb{R}^{2}$, where 
\begin{eqnarray*}
\mathcal{C}_{1} &=& \{x: b_{11}(x) \triangleq 9 - x_{1}^{2} - x_{2}^{2} \geq 0\} \\
&& \cap \{x: b_{12}(x) \triangleq -x_{1}-x_{2} \geq 0\} \\
\mathcal{C}_{2} &=& \{x: b_{21}(x) \triangleq 1-x_{1}^{3} \geq 0\} \\
&& \cap \{x: b_{22}(x) \triangleq x_{1}^{2} + x_{2}^{2} - 1 \geq 0\}
\end{eqnarray*}
Let $x_{1} = 0$, $x_{2} = 1$. We have that $x \notin \mathcal{C}_{1}$ with $b_{12}(x) < 0$, so $S = \{2\}$ and $j_{1} = 2$. Considering $\mathcal{C}_{2}$, we have $b_{21}(x) > 0$ and $b_{22}(x) = 0$, and hence $T_{2} = \{2\}$.

We let $\Delta(\mathcal{C})$ denote the set of $x \in \mathcal{C}$ that satisfy the conditions (i)--(iii) in Theorem \ref{prop:semialgebraic-boundary}. Theorem \ref{prop:semialgebraic-boundary} implies that $\partial \mathcal{C} \subseteq \Delta(\mathcal{C})$. The following describes the tangent cone at each point of $\Delta(\mathcal{C})$.

\begin{lemma}
\label{lemma:tangent-cone-semi-algebraic}
Suppose that $\mathcal{C}$ is a practical semi-algebraic set. Let $x \in \Delta(\mathcal{C})$. Let $S$ and $\{T_{i} : i \in S\}$ denote the maximal sets satisfying conditions (i)--(iii) of Theorem \ref{prop:semialgebraic-boundary} at $x$. Then $$\mathcal{T}_{\mathcal{C}}(x) = \bigcup_{i \in S}{\bigcap_{j \in T_{i}}{\left\{z: \frac{\partial b_{ij}}{\partial x}z \geq 0\right\}}}.$$
\end{lemma}

Lemma \ref{lemma:tangent-cone-semi-algebraic} implies that, if $\mathcal{C}$ is practical, then the tangent cone can be written as a union of the tangent cones of the simple semi-algebraic sets comprising $\mathcal{C}$. In the preceding example, we would have $\mathcal{T}_{\mathcal{C}}(x) = \{z : x_{1}z_{1} + x_{2}z_{2} \geq 0\}$. 

Intuitively, $x(t)$ must remain within the region $\mathcal{C}$ for all time. If $x(t) = (0 \ 1)$, then $x(t)$ is in the exterior of $\mathcal{C}_{1}$, and so the control must be chosen to ensure that $x(t)$ remains in $\mathcal{C}_{2}$. Furthermore, since $b_{21}(x) > 0$, it is only required to choose a control input at time $t$ to ensure that $x$ continues to satisfy $b_{22}(x) \geq 0$.

Before presenting our SOS approach to viability verification, we introduce the following preliminary result. This preliminary result uses Farkas lemma to prove that existence of a control $u \in \mathcal{U}$ satisfying the conditions of Nagumo's Theorem is equivalent to the non-existence of solutions to a system of polynomial inequalities.

%\begin{lemma}
%    \label{lemma:Psatz-Farkas}
 %   Let $\mathcal{P}$ and $\mathcal{Q}$ be sets of polynomials in $x_{1},\ldots,x_{n}$. Let $\Theta(x) \in \mathbb{R}^{N \times m}$ and $\psi(x) \in \mathbb{R}^{N}$ be polynomial functions of $x$. Then there is a solution to the linear inequalities $\Theta(x)u \leq \psi(x)$ for all $x$ with $p(x) > 0$ and $q(x) = 0$ if and only if there exist $\lambda \in \Lambda$, $\phi \in \Phi$, and $\sigma \in \Sigma$ with $\lambda + \phi^{2} + \sigma = 0$. In the above $\Lambda, \Phi, \Sigma \subseteq \mathbb{R}[x_{1},\ldots,x_{n},y_{1},\ldots,y_{N}]$ are defined by 
 %   \begin{IEEEeqnarray*}{rCl}
  %  \Lambda &=& \mathcal{I}[\Theta_1(x)^{T}y,\ldots,\Theta_m(x)^{T}y, q(x): q \in \mathcal{Q}] \\
  %  \Phi &=& \mathcal{M}[\psi(x)^{T}y, p(x): p \in \mathcal{P}] \\
   % \Sigma &=& \Sigma[y_{1},\ldots,y_{N}, -\psi(x)^{T}y, p(x): p \in \mathcal{P}]
   % \end{IEEEeqnarray*}
   % and $\Theta_{1}(x),\ldots,\Theta_{N}(x)$ are the columns of $\Theta(x)$.
%\end{lemma}

\begin{lemma}
\label{lemma:Psatz-Farkas}
Let $R$ be a positive integer, and suppose that  $\Theta_{1}(x),\ldots,\Theta_{R}(x) \in \mathbb{R}^{N \times m}$ and $\psi_{1}(x),\ldots,\psi_{R}(x) \in \mathbb{R}^{N}$ are polynomial functions of $x$. Let $\mathcal{P}$ and $\mathcal{Q}$ be finite sets of scalar polynomial functions of $x$. Then the following are equivalent:
\begin{enumerate}
    \item For any $x$ with $p(x) > 0$ for all $p \in \mathcal{P}$ and $q(x) = 0$ for all $q \in \mathcal{Q}$, there exist $i \in \{1,\ldots,R\}$ and $u \in \mathbb{R}^{m}$ such that $\Theta_{i}(x)u \leq \psi_{i}(x)$
    \item There exist polynomials $\lambda \in \Lambda$, $\phi \in \Phi$, and $\sigma \in \Gamma$ with $\lambda + \phi^{2} + \sigma = 0$, where $\Lambda, \Phi, \Gamma \subseteq \mathbb{R}[x_{1},\ldots,x_{n},y_{ij} : i=1,\ldots,R, j=1,\ldots,N]$ are defined by
    \begin{IEEEeqnarray*}{rCl}
        \Gamma &=& \Gamma[(y_{ij} : i=1,\ldots,R, j=1,\ldots,N),\\
        && (-\psi_{i}(x)^{T}y : i=1,\ldots,R), (p(x) : p \in \mathcal{P})]\\
    \Lambda &=& \mathcal{I}[(\Theta_{ij}(x)^{T}y_{i} : i=1,\ldots,R,j=1,\ldots,m), \\
    && (q(x): q \in \mathcal{Q})] \\
    \Phi &=& \mathcal{M}[(\psi_{i}(x)^{T}y_{i} : i=1,\ldots,R), (p(x) : p \in \mathcal{P})] 
    \end{IEEEeqnarray*}
    and $\Theta_{i1}(x),\ldots,\Theta_{im}(x)$ are the columns of $\Theta_{i}(x)$.
\end{enumerate}
\end{lemma}

We let $\mathcal{B}(\{\Theta_{1},\ldots,\Theta_{R}\},\{\psi_{1},\ldots,\psi_{R}\},\mathcal{P},\mathcal{Q})$ denote the set of polynomials $(\lambda,\eta,\sigma)$ satisfying the conditions of Lemma \ref{lemma:Psatz-Farkas}. For any $\Theta_{1},\ldots,\Theta_{R}$, $\psi_{1},\ldots,\psi_{R}$, $\mathcal{P}$, and $\mathcal{Q}$, we can check whether $\mathcal{B}(\{\Theta_{1},\ldots,\Theta_{R}\},\{\psi_{1},\ldots,\psi_{R}\},\mathcal{P},\mathcal{Q})$ is nonempty by solving the problem Psatz-SOS$(\mathcal{F},\mathcal{G},\mathcal{H})$, where
\begin{eqnarray*}
    \mathcal{F} &=& \{y_{ij} : i=1,\ldots,R, j=1,\ldots,N\} \\
    && \cup \{-\psi_{i}(x)^{T}y : i=1,\ldots,R\} \cup \{p(x) : p \in \mathcal{P}\} \\
    \mathcal{G} &=& \{\Theta_{ij}(x)^{T}y_{i} : i=1,\ldots,R,j=1,\ldots,m\} \\
    && \cup \{q(x) : q \in \mathcal{Q}\}\} \\
    \mathcal{H} &=& \{\psi_{i}(x)^{T}y_{i} : i=1,\ldots,R\} \cup \{p(x) : p \in \mathcal{P}\}
\end{eqnarray*}

Before stating the main result, we introduce some notation. For any finite set $\{1,\ldots,s\}$ and any $S \subseteq \{1,\ldots,s\}$, we let $\Pi(S)$ denote the set of mappings $\pi: S \rightarrow \mathbb{Z}_{+}$ with $\pi(i) \in \{1,\ldots,r_{i}\}$. We let $Z(S) = \prod_{i \in S}{2^{\{1,\ldots,r_{i}\}}}$, i.e., collections of subsets $(T_{i} : i \in S)$. 

Let $\mathcal{C}$ be a practical semi-algebraic set defined as in (\ref{eq:closed-semialgebraic}). For any $S \subseteq \{1,\ldots,s\}$, $T = (T_{i} : i \in S) \in Z(S)$, and $l \in S$ we define $\hat{\Theta}_{S,T,l}(x)$ to be a matrix with rows $-\frac{\partial b_{lj}}{\partial x}g(x)$ for $j \in T_{l}$, and define $\hat{\psi}_{S,T,l}(x)$ to be a vector with entries $\frac{\partial b_{lj}}{\partial x}f(x)$ for  $j \in T_{l}$.

We now state the main result on viability of semi-algebraic sets.

\begin{theorem}
    \label{theorem:semialgebraic-viability}
    Let $\mathcal{C}$ be a practical semi-algebraic set defined as in (\ref{eq:closed-semialgebraic}). Suppose that a system has dynamics (\ref{eq:dynamics}) with $f,g$ polynomials, and the set of inputs $u$ satisfies $\mathcal{U} = \{u: Au \leq c\}$. Then  $\mathcal{C}$ is viable if and only if for every $S \subseteq \{1,\ldots,s\}$, $\pi \in \Pi(\{1,\ldots,s\}\setminus S)$, and  $(T_{i} : i \in S) \in Z(S)$, the set $\mathcal{B}(\{\Theta_{S,T,l} : l \in S\}, \{\psi_{S,T,l} : l \in S\}, \mathcal{P}_{S,T,\pi}, \mathcal{Q}_{S,T,\pi})$
     is nonempty, where 
    \begin{IEEEeqnarray*}{rCl}
    \Theta_{S,T,l} &=& \left(
    \begin{array}{c}
    \hat{\Theta}_{S,T,l}(x) \\
    A
    \end{array}
    \right), \quad \psi_{S,T,l}(x) = \left(
    \begin{array}{c}
    \hat{\psi}_{S,T,l}(x) \\
    c
    \end{array}
    \right) \\
    \mathcal{P}_{S,T,\pi} &=& \{-b_{i\pi(i)} : i \notin S\} \cup \{b_{ij} : i \in S, j \notin T_{i}\} \\
    \mathcal{Q}_{S,T,\pi} &=& \{b_{ij}: i \in S, j \in T_{i}\}
    \end{IEEEeqnarray*}
\end{theorem}

From Theorem \ref{theorem:semialgebraic-viability}, the viability of a semi-algebraic set can be verified by solving a collection of SOS feasibility problems to verify that $\mathcal{B}(\{\Theta_{S,T,l} : l \in S\}, \{\psi_{S,T,l} : l \in S\}, \mathcal{P}_{S,T,\pi}, \mathcal{Q}_{S,T,\pi})$ is nonempty for all $S$, $T$, $l$, and $\pi$. 

The results of this section imply that, if the degree of the polynomials in the (Psatz-SOS) optimization are allowed to be arbitrarily large, then a solution is guaranteed to exist if and only if  $\mathcal{C}$ is viable. Note that, in practice, numerically solving the SOS programs will require limiting the degree of the polynomials in the optimization problem. Hence, while existence of a feasible solution implies safety of the system, non-existence of feasible solutions does not necessarily imply that the system is unsafe, only that there does not exist a safety certificate with the desired degree.

We observe that the number of SOS programs that must be solved as well as the dimension of each program may be large, especially for high-dimensional polynomials with high-degree functions $b_{ij}$. In the following section, we describe how to derive simplified SOS programs for special cases of semi-algebraic sets.

\subsection{Special Cases of the Semi-Algebraic Framework}
\label{subsec:special-cases}
 In what follows, we provide derivations for several such special cases of semi-algebraic sets. Each of these cases can be motivated by different types of CBF-based control. First, we consider the case where $s=1$ and $r_{1}=1$, i.e., when the set $\mathcal{C} = \{x : b(x) \geq 0\}$ is the super-level set of a single polynomial. This  corresponds to control policies with a single CBF constraint. 
 Second, we consider the case where $s=1$ and $r_{1} > 1$, i.e., the set $\mathcal{C}$ is a simple semi-algebraic set. This case corresponds to control policies with multiple CBF constraints, with High-Order Control Barrier Functions (HOCBFs) as a special case. 
 Finally, we consider the case where $s > 1$ and $r_{i} =1$ for $i=1,\ldots,s$. This case arises when there are multiple different CBF policies for different regions of the state space and the controller must switch between them.
\subsubsection{Super-Level Set of Single Polynomial} For simplicity, we first consider the case where $\mathcal{U} = \mathbb{R}^{m}$, meaning there are no constraints on the control input. We have the following result that follows  from Theorem \ref{theorem:semialgebraic-viability}.

\begin{corollary}
\label{prop:unconstrained-single-poly-viability}
Suppose that $\mathcal{C} = \{x : b(x) \geq 0\}$ for a polynomial $b$. Then $\mathcal{C}$ is viable if and only if $\mathcal{B}(0,0,\mathcal{P},\mathcal{Q})$ is empty, where 
\begin{IEEEeqnarray*}{rCl}
\mathcal{P} &=& \left\{-\frac{\partial b}{\partial x}f(x)\right\} \\
\mathcal{Q} &=& \left\{b,\frac{\partial b}{\partial x}g_{1}(x),\ldots,\frac{\partial b}{\partial x}g_{m}(x)\right\}
\end{IEEEeqnarray*}
\end{corollary}

\begin{proof}
    Since $s=1$, $r_{1}=1$, and there are no input constraints, we have that the conditions of Theorem \ref{theorem:semialgebraic-viability} only need to be checked for $S = T = \{1\}$. For simplicity, we omit the subscripts to obtain 
    \begin{IEEEeqnarray*}{rCl}
        \Theta &=& -\frac{\partial b}{\partial x}g(x), \quad \psi = \frac{\partial b}{\partial x}f(x) \\
        \mathcal{P} &=& \emptyset, \mathcal{Q} = \{b\}
    \end{IEEEeqnarray*}
    We can further simplify as follows. We know that, since $u$ is unconstrained and $\psi$ is a scalar, there is no solution $u$ to $\Theta(x) u \leq \psi$ if and only if $\Theta(x) = 0$ and $\psi$ is negative. Hence $\mathcal{C}$ is viable if and only if $\mathcal{B}(0,0,\mathcal{P},\mathcal{Q})$ as defined in the statement of the proposition is empty.
\end{proof}

We observe that verifying viability of $\mathcal{C}$ from Corollary \ref{prop:unconstrained-single-poly-viability} is equivalent to checking feasibility of the constraint 
\begin{multline*}
\eta(x)b(x) + \sum_{i=1}^{m}{\theta_{i}(x)\frac{\partial b}{\partial x}g_{i}(x)} - \alpha_{1}(x)\frac{\partial b}{\partial x}f(x) + \alpha_{0}(x) \\
+ \left(\frac{\partial b}{\partial x}f(x)\right)^{2v} = 0
\end{multline*}
for some integer $v$, polynomials $\eta(x),\theta_{1}(x),\ldots,\theta_{m}(x)$, and SOS polynomials $\alpha_{1}(x)$ and $\alpha_{0}(x)$. This is equivalent to a single SOS constraint.

In order to incorporate actuation constraints, we have the following generalization of Corollary \ref{prop:unconstrained-single-poly-viability}.

\begin{corollary}
\label{prop:constrained-single-poly-viability}
If the set $\mathcal{C} = \{x: b(x) \geq 0\}$ for a polynomial $b$ and $\mathcal{U} = \{u : Au \leq c\}$, then $\mathcal{C}$ is viable if and only if the set $\mathcal{B}(\Theta, \psi, \mathcal{P}, \mathcal{Q})$ is nonempty, where 
\begin{IEEEeqnarray*}{rCl}
\Theta &=& \left(
\begin{array}{c}
-\frac{\partial b}{\partial x}g(x) \\
A
\end{array}
\right), \quad \psi = \left(
\begin{array}{c}
\frac{\partial b}{\partial x}f(x) \\
c
\end{array}
\right) \\
\mathcal{P} &=& \emptyset, \quad \mathcal{Q} = \{b\}
\end{IEEEeqnarray*}
\end{corollary}

\begin{proof}
    The proof follows directly from Theorem \ref{theorem:semialgebraic-viability} with $s=1$ and $S = T = \{1\}$.
\end{proof}

\textbf{Example 2:} Consider the polynomial system
\begin{eqnarray*}
    \dot{x}_{1}(t) &=& -x_{1}(t)^{3} + u(t) \\
    \dot{x}_{2}(t) &=& x_{1} + 4x_{2}(t)^{2}
\end{eqnarray*}
with $\mathcal{C} = \{x : x_{1}^{2} - x_{2}^{3} + \beta \geq 0\}$ for some $\beta > 0$. Since $\frac{\partial b}{\partial x} = 0$ if and only if  $x_{1} = x_{2} = 0$, the set $\mathcal{C}$ is practical. We have $\frac{\partial b}{\partial x}g(x) = 2x_{1}$ and $\frac{\partial b}{\partial x}f(x) = -2x_{1}^{4} - 3x_{1}x_{2}^{2} - 12x_{2}^{2}$. Choosing $\eta(x) = 0$, $\alpha_{1}(x) = 12x_{2}^{2}$, and $\theta(x) = 12x_{1}^{3}x_{2}^{2} + 18x_{2}^{4} - 4.5x_{1}x_{2}^{4}$ yields
$$\eta(x) + \theta\frac{\partial b}{\partial x}g(x) - \alpha_{1}(x)\frac{\partial b}{\partial x}f(x) - \left(\frac{\partial b}{\partial x}f(x)\right)^{2} = 0 \in SOS,$$ implying that $\mathcal{C}$ is viable. 

Suppose that we add a constraint that $u \in \mathcal{U} = [-1,1]$. Then the problem would reduce to checking whether the set $\mathcal{B}(\Theta, \psi, \mathcal{P}, \mathcal{Q})$ is nonempty, where 
\begin{eqnarray*}
    \Theta &=& \left(
    \begin{array}{c}
    -2x_{1} \\
    1 \\
    -1
    \end{array}
    \right), \quad \psi = \left(
    \begin{array}{c}
    -2x_{1}^{4} - 3x_{1}x_{2}^{2} - 12x_{2}^{2} \\
    1 \\
    1
    \end{array}
    \right) \\
    \mathcal{P} &=& \emptyset, \quad \mathcal{Q} = \{x_{1}^{2} - x_{2}^{3} + \beta\}
\end{eqnarray*}
This can be verified by solving Psatz-SOS$(\mathcal{F},\mathcal{G},\mathcal{H})$, where 
\begin{eqnarray*}
\mathcal{F} &=& \{y_{1},y_{2}, y_{3}, y_{1}(2x_{1}^{4}+3x_{1}x_{2}^{2}+12x_{2}^{2}) -y_{2} - y_{3}\} \\
\mathcal{G} &=& \{-2x_{1}y_{1} + y_{2} - y_{3}, x_{1}^{2} - x_{2}^{3} + \beta\} \\
\mathcal{H} &=& \{y_{1}(2x_{1}^{4}+3x_{1}x_{2}^{2}+12x_{2}^{2}) -y_{2} - y_{3}\}
\end{eqnarray*}

\subsubsection{Simple Semi-Algebraic Sets} We next consider the case of simple semi-algebraic sets of the form $\mathcal{C} = \{x : b_{i}(x) \geq 0, i=1,\ldots,r\}$, where $b_{1},\ldots,b_{r}$ are polynomials. The following result applies Theorem \ref{theorem:semialgebraic-viability} to this case.

\begin{theorem}
\label{prop:simple-semialgebraic-viability}
Suppose that $\mathcal{C} = \{x : b_{i}(x) \geq 0, i=1,\ldots,r\}$ for some polynomials $b_{1},\ldots,b_{r}$. Then the set $\mathcal{C}$ is viable if and only if, for every $T=\{i_{1},\ldots,i_{l}\} \subseteq \{1,\ldots,r\}$, the set $\mathcal{B}(\theta_{T},\psi_{T},\mathcal{P}_{T},\mathcal{Q}_{T})$ is nonempty, where 
\begin{IEEEeqnarray}{rCl}
\label{eq:simple-semi-algebraic-1}
    \Theta_{T} &=& \left(
    \begin{array}{c}
    -\frac{\partial b_{i_{1}}}{\partial x}g(x) \\
    \vdots \\
    -\frac{\partial b_{i_{l}}}{\partial x}g(x) \\
    A
    \end{array}
    \right), \quad \psi_{T}(x) = \left(
    \begin{array}{c}
        \frac{\partial b_{i_{1}}}{\partial x}f(x) \\
        \vdots \\
        \frac{\partial b_{i_{l}}}{\partial x}f(x) \\
        c
    \end{array}
    \right) \\
    \label{eq:simple-semi-algebraic-2}
    \mathcal{P}_{T} &=& \{b_{i} : i \notin T\}, \quad \mathcal{Q}_{T} = \{b_{i} : i \in T\}
\end{IEEEeqnarray}
\end{theorem}

\begin{proof}
    When $s=1$, we only consider the case where $S = \{1\}$. Hence the set of permutations $\pi$ in Theorem \ref{theorem:semialgebraic-viability} is empty and $Z(\{1\}) = 2^{\{1,\ldots,r\}}$. Applying the conditions of Theorem \ref{theorem:semialgebraic-viability} (and omitting the $S$ subscript to simplify notations), we have that $\Theta_{T}$ and $\psi_{T}$ are defined as in (\ref{eq:simple-semi-algebraic-1}). Since the set of permutations $\pi$ is empty, $\mathcal{P}_{T}$ and $\mathcal{Q}_{T}$ reduce to (\ref{eq:simple-semi-algebraic-2}).
\end{proof}

\textbf{Example 3:} Suppose that $x(t) \in \mathbb{R}^{2}$ and 
\begin{eqnarray*}
\dot{x}_{1}(t) &=& -x_{1} + u \\
\dot{x}_{2}(t) &=& x_{1} + 4x_{2}
\end{eqnarray*}
Consider the set $\mathcal{C} = \{x : ||x||_{1} \leq 1\}$, where $||\cdot||_{1}$ denotes the $1$-norm. Suppose that there are no constraints on the control $u$, i.e., $\mathcal{U} = \mathbb{R}$. We have that $\mathcal{C} = \bigcap_{i=1}^{4}{\{x: b_{i}(x) \geq 0\}}$, where $b_{1}(x) = 1-x_{1}-x_{2}$, $b_{2}(x) = 1-x_{1} + x_{2}$, $b_{3}(x) = 1+x_{1}-x_{2}$, and $b_{4}(x) = 1+x_{1}+x_{2}$. Among all of the possible subsets T of $\{1,\ldots,4\}$ with $|T| > 1$, only four sets $T \subseteq \{1,\ldots,4\}$ have nonempty $\bigcap_{i \in T}{\{x : b_{i}(x) = 0\}}$. These sets are $\{1,2\}$, $\{1,3\}$, $\{2,4\}$, $\{3,4\}$. In order to ensure viability, the conditions of Proposition \ref{prop:simple-semialgebraic-viability} must be checked for all such sets $T$. As an example for $T = \{1,2\}$, we have 
\begin{eqnarray*}
\Theta_{T} &=& \left(
\begin{array}{c}
-1 \\
1
\end{array}
\right), \quad \psi_{T} = \left(
\begin{array}{c}
-4x_{2} \\
2x_{1} + 4x_{2}
\end{array}
\right) \\
\mathcal{P}_{T} &=& \{1 + x_{1}-x_{2}, 1+ x_{1} + x_{2}\}, \\ \mathcal{Q}_{T} &=& \{1-x_{1}-x_{2},1-x_{1}+x_{2}\}
\end{eqnarray*}
Note that, for this choice of $\mathcal{T}$, the set $\mathcal{B}(\Theta_{T},\psi_{T}, \mathcal{P}_{T}, \mathcal{Q}_{T})$ is nonempty, however, for $T = \{2,4\}$, the point $(0 \ 1)$ lies on the boundary $\{b_{2}(x) = 0\} \cap \{b_{4}(x) = 0\}$ and yet fails the conditions of Nagumo's Theorem, implying that $\mathcal{C}$ is not viable.

Proposition \ref{prop:simple-semialgebraic-viability} provides exact SOS conditions for verifying viability of a simple semi-algebraic set, however, the number of constraints grows exponentially in the number of functions $r$. An important special case of simple semi-algebraic sets that do \emph{not} experience exponential growth in computation is sets derived from High-Order Control Barrier Functions (HOCBFs)~\cite{xiao2021high}, which we define as follows. Recall that a nonlinear system 
\begin{IEEEeqnarray}{rCl}
\label{eq:HOCBF-dynamics-1}
\dot{x}(t) &=& f(x(t)) + g(x(t))u(t) \\
\label{eq:HOCBF-dynamics-2}
z(t) &=& h(x(t))
\end{IEEEeqnarray}
has relative degree $d$ if $L_{g}L_{f}^{i-1}h(x) = 0$ for $i=1,2,\ldots,(d-1)$ and $L_{g}L_{f}^{d-1}h(x) \neq 0$ for all $x$, where $L_{f}$ and $L_{g}$ denote the Lie derivatives of $h$ with respect to $f$ and $g$, respectively.

\begin{definition}[\cite{xiao2021high}]
    Suppose that the system $\dot{x}(t) = f(x(t)) + g(x(t))u(t)$ has relative degree $d$ with respect to a function $h$. The functions $\{b_{0},\ldots,b_{d-1}\}$ define a High-Order Control Barrier Function if
    \begin{IEEEeqnarray*}{rCl}
    b_{0}(x) &=& h(x) \\
    b_{i}(x) &=& L_{f}b_{i-1} + \kappa_{i}(b_{i-1}(x))
    \end{IEEEeqnarray*}
    for some class-$K$ functions $\kappa_{1},\ldots,\kappa_{d-1}$, where $L_{f}b_{i}$ is the Lie derivative of $b_{i}$ with respect to $f$.
\end{definition}

The following result gives equivalent conditions for verifying viability of HOCBF-defined semi-algebraic sets.

\begin{theorem}
\label{prop:HOCBF-viability}
Suppose that $b_{0},\ldots,b_{d-1}$ are polynomial functions that define an HOCBF for system (\ref{eq:HOCBF-dynamics-1})--(\ref{eq:HOCBF-dynamics-2}). Define matrix $\Theta(x) \in \mathbb{R}^{(p+1) \times m}$ and vector $\psi(x) \in \mathbb{R}^{p+1}$ by
\begin{equation}
\label{eq:HOCBF-statement}
    \Theta(x) = \left(
    \begin{array}{c}
    -\frac{\partial b_{d-1}}{\partial x}g(x) \\
    A
    \end{array}
    \right), \quad \psi(x) = \left(
    \begin{array}{c}
    \frac{\partial b_{d-1}}{\partial x}f(x) \\
    c
    \end{array}
    \right)
\end{equation}
Then the set $\mathcal{C} = \{x: b_{i}(x) \geq 0, i=0,\ldots,(d-1)\}$ is viable if and only if there exist polynomials $\lambda, \phi, \sigma \in \mathbb{R}[x_{1},\ldots,x_{n},y_{1},\ldots,y_{p+1}]$ that satisfy $\lambda + \phi^{2} + \sigma = 0$ and
\begin{IEEEeqnarray*}{rCl}
\lambda &\in& \mathcal{I}[\Theta_{1}(x)^{T}y,\ldots,\Theta_{m}(x)^{T}y, b_{d-1}] \\
\phi &\in& \mathcal{M}[\psi(x)^{T}y] \\
\sigma &\in& \Gamma[y_{1},\ldots,y_{p+1},-\psi(x)^{T}y,b_{0}(x),\ldots,b_{d-2}]
\end{IEEEeqnarray*}
\end{theorem}

\begin{proof}
Suppose that we apply the conditions of Proposition \ref{prop:simple-semialgebraic-viability} to the region $\mathcal{C}$. We have that the relative degree property and the definition of $b_{0},\ldots,b_{d-1}$ imply that the constraint $\frac{\partial b_{i}}{\partial x}(f(x)+g(x)u) \geq 0$ is automatically satisfied for $i=0,\ldots,(d-2)$ whenever $x \in \mathcal{C}$, and hence the corresponding rows of $\Theta_{T}$ and $\psi_{T}$ can be omitted, resulting in $\Theta_{T} = \Theta$ and $\psi_{T} = \psi$ for all $T \subseteq \{0,\ldots,(d-1)\}$, where $\Theta$ and $\psi$ are defined as in (\ref{eq:HOCBF-statement}). Furthermore, when $(d-1) \notin T$, the viability conditions are automatically satisfied since the set $\{u: Au \leq c\}$ is assumed to be nonempty.

Based on the above, we have that viability is satisfied if and only if, for any $x$ with $b_{d-1}(x) = 0$ and $b_{i}(x) \geq 0$ for $i=0,\ldots,(d-1)$, there exists $u$ with $\Theta(x)u \leq \psi(x)$. Equivalently, for any $x$ with $b_{d-1}(x) =0$ and $b_{i}(x) \geq 0$ for $i=0,\ldots,(d-2)$, there does not exist $y \in \mathbb{R}^{p+1}_{\geq 0}$ with $\Theta(x)^{T}y = 0$ and $\psi(x)^{T}y < 0$. Viability of the HOCBF is therefore satisfied if and only if there do not exist $x \in \mathbb{R}^{n}$ and $y \in \mathbb{R}^{p+1}$ with $y_{1},\ldots,y_{p+1} \geq 0$, $b_{i}(x) \geq 0$ for $i=0,\ldots,(d-2)$, $b_{d-1}(x) = 0$, $\Theta(x)^{T}y=0$, and $\psi(x)^{T}y < 0$. These conditions are equivalent to the conditions of the proposition by Theorem \ref{theorem:Psatz}.
\end{proof}

%\textbf{Remark that high relative degree property can also be checked in a semialgebraic manner.}

\subsubsection{Union of Polynomial Super-Level Sets}
Finally, we consider the case where the region $\mathcal{C}$ is defined by a union of polynomial super-level sets, i.e., $$\mathcal{C} = \bigcup_{i=1}^{s}{\{x: b_{i}(x) \geq 0\}}.$$ This is the special case where $s > 1$ and $L_{i} = 1$ for $i=1,\ldots,s$. We have the following result.

\begin{corollary}
\label{prop:union-viability}
Suppose that $\mathcal{C} = \bigcup_{i=1}^{s}{\{x : b_{i}(x) \geq 0\}}$ for some polynomials $b_{1},\ldots,b_{s}$. Then the set $\mathcal{C}$ is viable if and only if $\mathcal{B}(\{\Theta_{l} : l \in S\}, \{\psi_{l} : l \in S\}, \mathcal{P}_{S}, \mathcal{Q}_{S})$ is nonempty for all $S \subseteq \{1,\ldots,s\}$, where
\begin{IEEEeqnarray*}{rCl}
\Theta_{l} &=& \left(
\begin{array}{c}
-\frac{\partial b_{l}}{\partial x}g(x) \\
A
\end{array}
\right), \quad \psi_{l} = \left(
\begin{array}{c}
\frac{\partial b_{l}}{\partial x}f(x) \\
c
\end{array}
\right) \\
\mathcal{P}_{S} &=& \{-b_{i} : i \notin S\}, \quad \mathcal{Q}_{S} = \{b_{i} : i \in S\}
\end{IEEEeqnarray*}
\end{corollary}

\begin{proof}
    The proof follows from Theorem \ref{theorem:semialgebraic-viability} and the fact that, when $L_{i}=1$ for $i=1,\ldots,s$, there is a unique mapping $\pi$ given by $\pi(i) = 1$ for all $i \notin S$ and the set $T_{i} = \{b_{i}\}$ for all $i \in S$.
\end{proof}

\subsection{Verifying Inclusion in a Given Safe Region}
\label{subsec:safe-inclusion}
Safety constraints are typically expressed in the form $x(t) \in \mathcal{S}$ for all $t$, where $\mathcal{S}$ is a region of the state space. In order to ensure that safety requirements are met, we must verify that the invariant set $\mathcal{C}$ is contained in $\mathcal{S}$. We describe a procedure for verifying such conditions when $\mathcal{C}$ is semi-algebraic and $\mathcal{S} = \{x : h(x) \geq 0\}$ for some polynomial $h(x)$. 

\begin{lemma}
    \label{lemma:containment}
Let $\mathcal{C} = \bigcup_{i=1}^{s}{\bigcap_{j=1}^{r_{i}}{\{x: b_{ij}(x) \geq 0\}}}$.    The set $\mathcal{C} \subseteq \mathcal{S}$ if and only if, for each $i=1,\ldots,s$, there does not exist $x$ satisfying $b_{ij}(x) \geq 0$ for all $j=1,\ldots,r_{i}$ and $h(x) < 0$.
\end{lemma}

The proof is straightforward. As a result, we can obtain a set of $s$ SOS programs for verifying inclusion.

\begin{proposition}
\label{prop:containment}
The set $\mathcal{C}$ is contained in $\mathcal{S}$ if and only if, for all $i=1,\ldots,s$, there exists $\sigma \in \Gamma[b_{i1},\ldots,b_{ir_{i}},-h]$ and an integer $l$ such that $$\sigma + h^{2l} = 1.$$ 
\end{proposition}

The proof follows directly from Theorem \ref{theorem:Psatz}.

\subsection{Computational Complexity  of Problem \ref{problem:viability-verification}}
\label{subsec:complexity}
The characterization of viability for a single polynomial level set enables us to analyze the complexity of Problem \ref{problem:viability-verification} as follows.

\begin{theorem}
\label{prop:CBF-complexity}
Problem \ref{problem:viability-verification} is NP-hard.
\end{theorem}

\begin{proof}
If there is a polynomial-time algorithm for Problem \ref{problem:viability-verification}, then the algorithm can be used to check the conditions of Corollary \ref{prop:unconstrained-single-poly-viability}, since viability of a set $\{x: b(x) \geq 0\}$ for some polynomial $b(x)$ is a special case of viability of a practical semialgebraic set $\mathcal{C}$. Hence, in order to prove that Problem \ref{problem:viability-verification} is NP-hard, it suffices to prove that checking the conditions of Corollary \ref{prop:unconstrained-single-poly-viability} is NP-hard. We will show that there exists a polynomial time reduction of the NP-hard problem of checking nonnegativity of a polynomial to the problem of verifying conditions of Corollary \ref{prop:unconstrained-single-poly-viability}. 

Let $r: \mathbb{R}^{n} \rightarrow \mathbb{R}$ be a polynomial with degree four or more. Define $b: \mathbb{R}^{n+1} \rightarrow \mathbb{R}$ by $b(x) = x_{n+1}r(x_{1},\ldots,x_{n})$ where $x = (x_{1},\ldots,x_{n},x_{n+1})^{T}$. Define $g: \mathbb{R}^{n+1} \rightarrow \mathbb{R}^{(n+1) \times m}$ by $g(x) \equiv 0$. Define $f: \mathbb{R}^{n+1} \rightarrow \mathbb{R}^{n+1}$ as $f(x) = (f_{1}(x),\ldots,f_{n+1}(x)^{T}$ with 
\begin{displaymath}
f_{i}(x) = \left\{
\begin{array}{ll}
0, & i \neq n+1 \\
1, & i = n+1
\end{array}
\right.
\end{displaymath}

We have that $\frac{\partial b}{\partial x}g(x) = 0$ and $\frac{\partial b}{\partial x}f(x) = \frac{\partial b}{\partial x_{n+1}}f_{n+1}(x) = r(x_{1},\ldots,x_{n})$. Suppose that there exists $x$ such that $b(x) = 0$ and $\frac{\partial b}{\partial x}f(x) < 0$. Then $x_{n+1}r(x_{1},\ldots,x_{n}) = 0$ and $r(x_{1},\ldots,x_{n}) < 0$, implying that $r$ is \emph{not} nonnegative. Conversely, if no such $x$ exists, then in particular there is no $(x_{1},\ldots,x_{n})^{T}$ with $b(x_{1},\ldots,x_{n},0) = 0$ and $\frac{\partial b}{\partial x}f(x) = r(x_{1},\ldots,x_{n}) < 0$, implying that $r$ is nonnegative. Hence $r$ is nonnegative if and only if $b$ is  CBF, implying NP-hardness of Problem \ref{problem:viability-verification}.
\end{proof}

\subsection{Verifying Feedback Controlled Positive Invariance}
\label{subsec:feedback-control}
The preceding subsections considered viability analysis of semialgebraic sets, however, they did not establish existence of continuous control policies that ensure positive invariance. In this section, we present sufficient conditions for feedback controlled positive invariance. We  consider the case of a simple semi-algebraic set.

\begin{theorem}
\label{theorem:simple-CPI}
Suppose that $\mathcal{C} = \{x: b_{i}(x) \geq 0 \ \forall i=1,\ldots,r\}$ for some polynomials $b_{1},\ldots,b_{r}$ and that $\mathcal{C}$ is closed,  bounded, and locally compact. Suppose further that, for any $T \subseteq \{1,\ldots,r\}$ and $x$ with $b_{i}(x) = 0$ for $i \in T$, there exists $u$ such that $Au \leq c$ and $\frac{\partial b_{i}}{\partial x}(f(x)+g(x)u) > 0$ for all $i \in S$. Then $\mathcal{C}$ is feedback controlled positive invariant.
\end{theorem}

A proof can be found in the appendix. A similar approach to Section \ref{subsec:viability-verification} can be used to verify feedback controlled positive invariance through sum-of-squares optimization. 

\begin{theorem}
\label{prop:CPI-conditions}
Let $\mathcal{C} = \bigcap_{i=1}^{r}{\{x : b_{i}(x) \geq 0\}}$ be a simple practical semi-algebraic set. For any $S \subseteq \{1,\ldots,r\}$, define $\Theta_{S}(x)$ to be a matrix with rows $-\frac{\partial b_{i}}{\partial x}g(x)$ and $\psi_{S}(x)$ to be a vector with entries $\frac{\partial b_{i}}{\partial x}f(x)$ for $i \in S$. Suppose that, for all $S \subseteq \{1,\ldots,r\}$, the following conditions hold:
\begin{itemize}
    \item There do not exist $x \in \mathbb{R}^{n}$, $y \in \mathbb{R}^{|S|}$, and $z \in \mathbb{R}^{p}$ satisfying (i) $\Theta_{S}(x)^{T}y + A^{T}z = 0$, (ii) $\psi_{S}(x)^{T}y + c^{T}z < 0$, (iii) $b_{i}(x) = 0$ for all $i \in S$, and (iv) $b_{i}(x) > 0$ for all $i \in \{1,\ldots,r\} \setminus S$.
    \item There do not exist $x \in \mathbb{R}^{n}$, $y \in \mathbb{R}^{|S|}$, and $z \in \mathbb{R}^{p}$ satisfying (i) $\Theta_{S}(x)^{T}y + A^{T}z = 0$, (ii) $\psi_{S}(x)^{T}y + c^{T}z <\leq 0$, (iii) $b_{i}(x) = 0$ for all $i \in S$,  (iv) $b_{i}(x) > 0$ for all $i \in \{1,\ldots,r\} \setminus S$, and (v) $y \neq 0$
\end{itemize}
Then $\mathcal{C}$ is feedback controlled positive invariant.
\end{theorem}

A proof can be found in the appendix. We next turn to proving the existence of CBF-based control laws. 

\begin{theorem}
\label{theorem:CBF-policy-existence}
Suppose that the set $\mathcal{C} = \{x: b_{i}(x) \geq 0, i=1,\ldots,r\}$ is compact and the functions $b_{1},\ldots,b_{r}$ satisfy the conditions of Proposition \ref{prop:CPI-conditions}. Then there exist class-K functions $\kappa_{1},\ldots,\kappa_{r}: \mathbb{R} \rightarrow \mathbb{R}$ such that, for any function $J: \mathbb{R}^{n} \times \mathcal{U} \rightarrow \mathbb{R}$, the policy
\begin{equation}
    \label{eq:CBF-policy}
    \begin{array}{ll}
    \mbox{minimize} & J(x,u) \\
    \mbox{s.t.} & \frac{\partial b_{i}}{\partial x}(f(x) + g(x)u) \geq -\kappa_{i}(b_{i}(x)), i=1,\ldots,r \\
    & Au \leq c
    \end{array}
\end{equation}
renders $\mathcal{C}$ positive invariant, provided the minimum value of (\ref{eq:CBF-policy}) exists.
\end{theorem}

\begin{proof}
    It suffices to show that there exist $\kappa_{1},\ldots,\kappa_{r}$ such that for any $x \in \mathcal{C}$, there exists $u$ satisfying the constraints
    \begin{IEEEeqnarray}{rCl}
    \label{eq:CBF-policy-proof-1}
    \frac{\partial b_{1}}{\partial x}(f(x) + g(x)u) &\geq& -\alpha_{1}(b_{1}(x)) \\
    &\vdots & \\
    \frac{\partial b_{r}}{\partial x}(f(x) + g(x)u) &\geq& -\alpha_{r}(b_{r}(x)) \\
    \label{eq:CBF-policy-proof-2}
    Au &\leq& c 
    \end{IEEEeqnarray}
    If the conditions of Proposition \ref{prop:CPI-conditions} are  satisfied, there exists a continuous feedback control policy $\mu: \mathbb{R}^{n} \rightarrow \mathcal{U}$ that renders $\mathcal{C}$ positive invariant. Our approach will be to show that we can construct $\kappa_{1},\ldots,\kappa_{r}$ such that $\mu(x)$ is a solution to (\ref{eq:CBF-policy-proof-1})--(\ref{eq:CBF-policy-proof-2}) for all $x \in \mathcal{C}$. 

    Define $q_{i}(x) = -\frac{\partial b_{i}}{\partial x}(f(x)+g(x)\mu(x))$. Since $\mathcal{C}$ is compact, there exists $K_{i}$ such that $q_{i}(x) \leq K_{i}$ for all $x \in \mathcal{C}$. By construction of $\mu$, we have that $\frac{\partial b_{i}}{\partial x}(f(x) + g(x)\mu(x)) > 0$ for all $x$ with $b_{i}(x) = 0$. Hence, if we define $$z_{i} = \inf{\{b_{i}(x) : q_{i}(x) = 0, x \in \mathcal{C}\}},$$ we have that $z_{i} > 0$ by compactness of $\mathcal{C}$ and the fact that $\{x : q_{i}(x) = 0\}$ is closed. Combining these ideas, we can choose 
    \begin{displaymath}
    \kappa_{i}(z) = \left\{
    \begin{array}{ll}
    \frac{K_{i}}{z_{i}}z, & z \in [0,z_{i}] \\
    K_{i} + \epsilon(z-z_{i}), & z > z_{i}
    \end{array}
    \right.
    \end{displaymath}
    Clearly $\kappa_{i}$ is class-K. If $b_{i}(x) \in [0,z_{i}]$, then $$\frac{\partial b_{i}}{\partial x}(f(x) + g(x)\mu_{i}(x)) = -q_{i}(x) \geq 0 \geq -\kappa_{i}(b_{i}(x)).$$ If $b_{i}(x) > z_{i}$, then $$\frac{\partial b_{i}}{\partial x}(f(x) + g(x)\mu_{i}(x)) \geq -K_{i} > -(K_{i} + \epsilon(z-z_{i})) = -\kappa_{i}(b_{i}(x)).$$ Hence this choice of $\kappa_{i}$ ensures that (\ref{eq:CBF-policy-proof-1})--(\ref{eq:CBF-policy-proof-2}) hold everywhere in $\mathcal{C}$, completing the proof.
\end{proof}

\subsection{Generalization to Trigonometric Dynamics}
\label{subsec:trig-functions}
We next show that our approach can be generalized to a class of systems with trigonometric functions in the constraints and dynamics. We make the following assumptions regarding the dynamics.

\begin{assumption}
\label{assumption:trig}
Consider a system (\ref{eq:dynamics}) and semialgebraic set $\mathcal{C}$ such that there exist $l \in \{1,\ldots,n\}$ such that only trigonometric functions of the states $x_{l+1},\ldots,x_{n}$ appear in the dynamics and safety constraints. More precisely, let $L = 2(n-l)$ and define $w_{1},\ldots,w_{L}$ by $w_{2k-1} = \sin{x_{l+k}}$ and $w_{2k} = \cos{x_{l+k}}$. We assume that there exist polynomials $\overline{f}: \mathbb{R}^{l+L} \rightarrow \mathbb{R}^{n}$, $\overline{g}: \mathbb{R}^{l+L} \rightarrow \mathbb{R}^{n \times m}$, and $\overline{b}_{ij} : \mathbb{R}^{l+L} \rightarrow \mathbb{R}$ for $i=1,\ldots,s$, $j=1,\ldots,r_{i}$ such that 
\begin{IEEEeqnarray*}{rCl}
f(x) &=& \overline{f}(x_{1},\ldots,x_{l},w_{1},\ldots,w_{L}) \\
g(x) &=& \overline{g}(x_{1},\ldots,x_{l},w_{1},\ldots,w_{L}) \\
b_{ij}(x) &=& \overline{b}_{ij}(x_{1},\ldots,x_{l},w_{1},\ldots,w_{L})
\end{IEEEeqnarray*}
\end{assumption}
Our approach will be to convert the system (\ref{eq:dynamics}) to a polynomial system with an extended state space so that the techniques of this section can be used for verifying safety. We define a system with state $(\hat{x}(t),\hat{w}(t))$ where $\hat{x}(t) \in \mathbb{R}^{l}$ and $\hat{w}(t) \in \mathbb{R}^{2(n-l)}$ by 
\begin{IEEEeqnarray}{rCl}
\label{eq:trig-equiv-1}
\dot{\hat{x}}_{i}(t) &=& \overline{f}_{i}(\hat{x},\hat{w}), \ i=1,\ldots,l \\
\dot{\hat{w}}_{2k-1} &=& \hat{w}_{2k}\overline{f}_{l+k}(\hat{x},\hat{w}), \ k=1,\ldots,(n-l) \\
\label{eq:trig-equiv-2}
\dot{\hat{w}}_{2k} &=& -\hat{w}_{2k-1}\overline{f}_{l+k}(\hat{x},\hat{w}), k=1,\ldots,(n-l)
\end{IEEEeqnarray}
We then have the following result.

\begin{theorem}
\label{theorem:trig-verify}
Let $\mathcal{C}$ be defined by $$\mathcal{C} = \bigcup_{i=1}^{s}{\bigcap_{j=1}^{r_{i}}{\{x : b_{ij}(x) \geq 0\}}}$$ where the $b_{ij}$'s satisfy Assumption \ref{assumption:trig}. Define $\hat{\mathcal{C}}$ by 
\begin{multline*}
\hat{\mathcal{C}} \triangleq \left(\bigcup_{i=1}^{s}{\bigcap_{j=1}^{r_{i}}{\{(\hat{x},\hat{w}) : \overline{b}_{ij}(\hat{x},\hat{w}) \geq 0\}}}\right) \\
\cap \left(\bigcap_{k=1}^{n-1}{\{(\hat{x},\hat{w}) : \hat{w}_{2k-1}^{2} + \hat{w}_{2k}^{2} = 1\}}\right).
\end{multline*}
Then $\mathcal{C}$ is positive invariant under (\ref{eq:dynamics}) if and only if $\hat{\mathcal{C}}$ is positive invariant under (\ref{eq:trig-equiv-1})--(\ref{eq:trig-equiv-2}).
\end{theorem}

\begin{proof}
    Suppose first that $\mathcal{C}$ is not positive invariant, and let $x(0) \in \mathcal{C}$ be the initial state of a trajectory that exits $\mathcal{C}$. Let $\hat{x}(0)$ be defined by $\hat{x}_{i}(0) = x_{i}(0)$ for $i=1,\ldots,l$, $\hat{w}_{2k-1}(0) = \sin{x_{l+k}(0)}$, and $\hat{w}_{2k}(0) = \cos{x_{l+k}(0)}$. We then have $(\hat{x}(0),\hat{w}(0)) \in \hat{\mathcal{C}}$. Furthermore, by construction of (\ref{eq:trig-equiv-1})--(\ref{eq:trig-equiv-2}), we have that $w_{2k-1}(t) = \sin{x_{l+k}(t)}$ and $w_{2k} = \cos{x_{l+k}(t)}$ for all $t \geq 0$. Hence if $x(t) \notin \mathcal{C}$ for some time $t$, then $(\hat{x}(t),\hat{w}(t)) \notin \hat{\mathcal{C}}$.

    Now, suppose that $\hat{\mathcal{C}}$ is not positive invariant, and let $(\hat{x}(0), \hat{w}(0)) \in \hat{\mathcal{C}}$ be an initial state such that $(\hat{x}(t),\hat{w}(t)) \notin \hat{\mathcal{C}}$ for some time $t$. For all $k$, we have that 
    \begin{IEEEeqnarray*}{rCl}
    \IEEEeqnarraymulticol{3}{l}{
    \frac{d}{dt}(w_{2k-1}^{2} + w_{2k}^{2})} \\
    &=& 2\hat{w}_{2k-1}\hat{w}_{2k}\overline{f}_{k+l}(\hat{x},\hat{w}) - 2\hat{w}_{2k}\hat{w}_{2k-1}\overline{f}_{k+l}(\hat{x},\hat{w}) = 0,
    \end{IEEEeqnarray*}
    and hence we must have $$(\hat{x}(t),\hat{w}(t)) \notin \bigcup_{i=1}^{s}{\bigcap_{j=1}^{r_{i}}{\{(\hat{x},\hat{w}) : \overline{b}_{ij}(\hat{x},\hat{w}) \geq 0\}}}.$$ Choosing $x(0)$ such that $x_{i}(0) = \hat{x}_{i}(0)$ for $i=1,\ldots,l$, $\sin{x_{k+l}(0)} = \hat{w}_{2k-1}(0)$, and $\cos{x_{k+l}(0)} = \hat{w}_{2k}(0)$ for $k=1,\ldots,(n-l)$, we have that $\cos{x_{k}(t)} = \hat{w}_{2k}(t)$ and $\sin{x_{k}(t)} = \hat{w}_{2k-1}(t)$ for all $t \geq 0$. We then have $x(t) \notin \mathcal{C}$, completing the proof.
\end{proof}

%We now turn to the more general case of a semialgebraic set. Unfortunately, in this case even if the viability conditions are satisfied, we cannot guarantee the existence of a state feedback controller $\mu: \mathcal{C} \rightarrow \mathcal{U}$ that ensures positive invariance. As an example, consider a two-dimensional system with $\dot{x}(t) = u(t) + (1 \ 0)^{T}$ and $\mathcal{C} = \{x : b_{1}(x) \geq 0\} \cup \{x : b_{2}(x) \geq 0\}$, where $b_{1}(x) = 1-((x_{1}+1)^{2} + x_{2}^{2})$ and $b_{2}(x) = 1-((x_{1}-1)^{2} + x_{2}^{2})$. At $x=0$, we have that the control input must satisfy $-2u_{1} \geq 0$

\subsection{Applications of Our Approach}
\label{subsec:verification-applications}
We first motivate our verification framework by discussing how verifying safety of some previously proposed safe control methodologies can be viewed as special cases of our approach.

\subsubsection{Polynomial CBF and HOCBF} CBFs have seen widespread applications in autonomous driving \cite{chen2017obstacle}, aerospace \cite{breeden2021guaranteed}, and robotics \cite{kurtz2021control}. In many of these applications, the safety constraints can be encoded as polynomial constraints on the state space. These polynomial constraints are then used as control barrier functions, for example, the linear constraint in the adaptive cruise control scenario of \cite{ames2016control}, or the minimum-distance constraint of \cite{chen2017obstacle}. In this case, the set $\{x: b(x) \geq 0\}$, where $b(x)$ denotes the CBF, is semi-algebraic, and verifying that the CBF-based control policy ensures safety is equivalent to verifying invariance of this set. HOCBFs have been proposed as CBF constructions for high relative degree systems, and can also be verified using the approach of this paper (see Proposition \ref{prop:HOCBF-viability}).

\subsubsection{Neural Control Barrier Functions} A Neural Control Barrier Function (NCBF) is a CBF $b$ that is represented by a feedforward neural network \cite{dawson2023safe}. For certain widely-used activation functions such as ReLU and sigmoid, the set $\mathcal{C} = \{x: b(x) \geq 0\}$ is equal to a union of linear polytopes~\cite{zhang2024exact}, and hence is semi-algebraic. 

\subsubsection{Value Function Approximation} One line of work attempts to compute CBFs for a system by solving the Hamilton Jacobi Bellman equation to obtain a reach-avoid value function, which can then be used as a  CBF \cite{wabersich2023data}. In practice, since solving the HJB equation involves discretizing the state space and using polynomial interpolation, the approximate CBF $b$ defines a semi-algebraic set. We present two specific examples as follows.

In \cite{tonkens2022refining}, the authors present an iterative algorithm for computing the reach-avoid value function, and prove that it converges to a viable set. The algorithm for estimating the value function first discretizes the state space as a grid, solves a discrete HJB equation at each grid point, and computes an estimated barrier function as the interpolation of the value function.  Suppose that the state space is partitioned into a collection of hyperrectangles $\mathcal{R}_{i}$ where $\mathcal{R}_{i}$ has vertices $x_{i1},\ldots,x_{iM}$. Suppose that, on each hypperrectangle, the CBF $b$ is approximated by $h_{i}(x,V_{i1},\ldots,V_{iM})$, where $h_{i}$ is a polynomial and $V_{i1},\ldots,V_{iM}$ are the approximations of the value function at $x_{i1},\ldots,x_{iM}$ (for example, linear interpolation is used in \cite{tonkens2022refining}).  The safe region to be verified is then given by $$\mathcal{C} = \bigcup_{i=1}^{M}{\left\{\mathcal{R}_{i} \cap \{x: h_i(x,V_{1},\ldots,V_{M}) \geq 0\}\right\}},$$ which is semi-algebraic.
%Since this discrete approximation of the barrier function is what will eventually be used to control the system, it should be verified to ensure it guarantees safety in spite of any approximation errors due to discretization.

In \cite{kumar2023fast}, an alternative approach to value function-based synthesis was proposed. In this approach, a discrete-time model $x_{t+1} = \hat{f}(x_{t},u_{t})$ is used.   A trajectory $\overline{x}_{0},\ldots,\overline{x}_{T} \in \mathbb{R}^{n}$ and a sequence of inputs $\overline{u}_{0},\ldots,\overline{u}_{T} \in \mathbb{R}^{m}$ are computed, with $\overline{x}_{t+1} = \hat{f}(\overline{x}_{t},\overline{u}_{t})$ for $t=0,\ldots,(T-1)$ and $\overline{x}_{T}$ lying in a known invariant set. The reach-avoid value function $V$ is then approximated using differential dynamic programming, resulting in a set of quadratic approximations $$V_{t}(x) = \overline{V}_{t} + z_{t}^{T}(x-\overline{x}_{t}) + \frac{1}{2}(x-\overline{x}_{t})^{T}Z_{t}(x-\overline{x}_{t})$$ for some $\overline{V}_{t} \in \mathbb{R}$, $z_{t} \in \mathbb{R}^{n}$, and $Z_{t} \in \mathbb{R}^{n \times n}$. The set $\mathcal{C} = \bigcup_{t=0}^{T}{\{x : V_{t}(x) \geq 0\}}$ then serves as a candidate invariant set. Since this set is a union of super-level sets of quadratic functions, it is semi-algebraic. Viability of this set under a continuous-time model of the form (\ref{eq:dynamics}) can then be performed.

%In this approach, a safe path to a known invariant set is computed using differential dynamic programming, and a sequence of quadratic approximations are computed for the value function at points along this path. Letting these approximations be denoted by $b_{1}(x),\ldots,b_{N}(x)$, the approximate safe region is given by the semialgebraic set $\mathcal{C} = \bigcup_{i=1}^{N}{\{x : b_{i}(x) \geq 0\}}$.
\section{Problem Formulation: CBF Synthesis}
\label{sec:synthesis}
This section considers the following problem.
\begin{problem}
\label{problem:synthesis}
Given a safety constraint $\mathcal{S} = \{x: h(x) \geq 0\}$, construct a CBF $b: \mathbb{R}^{n} \rightarrow \mathbb{R}$  that guarantees positive invariance and $\mathcal{C} \subseteq \mathcal{S}$.
\end{problem}

We present two approaches to synthesizing CBFs. The first approach uses an alternating descent method. The second approach exploits the existence of equilibria to construct CBFs for local neighborhoods. %The third approach uses Newton's method in function space to guarantee local convergence to a CBF.

\subsection{Alternating Descent Approach}
\label{subsec:ad-synthesis}
Our first approach is based on the SOS constraints introduced in Section \ref{sec:verification}. When the function $b(x)$ is unknown, the constraints become  non-convex. We therefore propose an alternating-descent heuristic, in which we alternate between updating the CBF candidate $b(x)$ and the polynomials used for verification. We describe our approach for synthesizing CBFs for polynomial dynamics, although a similar approach can be exploited for systems with  trigonometric functions in their dynamics.

We initialize a function $b^{0}(x)$ arbitrarily (e.g., $b^{0}(x) = h_{i}(x)$ for some $i$) and initialize parameters $\rho_{0}$ and $\rho_{0}^{\prime}$ to be infinite. At step $k$, we solve the following SOS program with variables $\rho$, $\alpha^{k}(x)$, $\eta^{k}(x)$, $\theta_{1}^{k}(x),\ldots,\theta_{m}^{k}(x)$, $w(x)$, and $\beta^{k}(x)$:

\begin{equation}
\label{eq:alternating-descent-1}
\begin{array}{ll}
\mbox{minimize} & \rho \\
\mbox{s.t.} & \alpha^{k}(x)\frac{\partial b^{k-1}}{\partial x}f(x) + \sum_{i=1}^{m}{\theta_{i}^{k}(x)\left[\frac{\partial b^{k-1}}{\partial x}g(x)\right]_{i}} \\
& + \eta^{k}(x)b^{k-1}(x) + \rho \Lambda(x) - 1 \in SOS \\
& -\beta^{k}(x)b^{k-1}(x) + w(x)h(x) - 1 \in SOS \\
& \beta^{k}(x), \alpha^{k}(x) \in SOS
\end{array}
\end{equation}

In (\ref{eq:alternating-descent-1}), $\Lambda(x)$ is an SOS polynomial chosen to ensure that the first constraint is SOS for sufficiently large $\rho$. We let $\rho_{k}$ denote the value of $\rho$ returned by solving (\ref{eq:alternating-descent-1}). The procedure terminates if $\rho_{k} \leq 0$ or if $|\rho_{k} - \rho_{k-1}^{\prime}| < \epsilon$. Otherwise, we solve the SOS problem with variables $\rho$ and $b^{k}$ given by

\begin{equation}
\label{eq:alternating-descent-2}
\begin{array}{ll}
\mbox{minimize} & \rho \\
\mbox{s.t.} & \alpha^{k}(x)\frac{\partial b^{k}}{\partial x}f(x) + \sum_{i=1}^{m}{\theta_{i}^{k}(x)\left[\frac{\partial b^{k}}{\partial x}g(x)\right]_{i}} \\
& + \eta^{k}(x)b^{k}(x) + \rho\Lambda(x)-1 \in SOS \\
& -\beta^{k}(x)b^{k}(x) + w(x)h(x) - 1 \in SOS 
\end{array}
\end{equation}
We let $\rho_{k}^{\prime}$ denote the value of $\rho$ returned by the optimization. The procedure terminates if $\rho \leq 0$ or if $|\rho_{k}-\rho_{k}^{\prime}| < \epsilon$. A pseudocode description is given as Algorithm \ref{algo:synthesis}.
\begin{algorithm}[h]
\caption{CBF Synthesis}
\begin{algorithmic}[1]
    \State \textbf{Input:} Initial candidate CBF $b^{0}$, dynamics functions $f$ and $g$, safety constraint $h$, tolerance $\epsilon$.
    \State \textbf{Output:} Either a valid CBF $b$ or \textbf{fail} if CBF cannot be found
    \Procedure{CBF Synthesis}{$b^{0},f,g,h,\epsilon$}
    \State $k \leftarrow 0$
    \State $\rho_{0} \leftarrow 1$, $\rho_{0}^{\prime} \leftarrow (1+2\epsilon)$ 
    \While{$\rho_{k} > 0$ \textbf{and} $\rho_{k}^{\prime} > 0$ \textbf{and} $|\rho_{k}-\rho_{k}^{\prime}| > \epsilon$}
    \State $(\rho_{k},\alpha^{k},\theta_{1}^{k},\ldots,\theta_{m}^{k},\eta^{k},\beta^{k}) \rightarrow$ solution to \eqref{eq:alternating-descent-1}
    \State $(\rho_{k}^{\prime},b^{k})\leftarrow $ solution to \eqref{eq:alternating-descent-2}
    \EndWhile
    \If{$\rho_{k} \leq 0$ \textbf{or} $\rho_{k}^{\prime} \leq 0$}
    \State \Return{$b^{k}$}
    \Else
    \State \Return{\textbf{fail}}
    \EndIf
   \EndProcedure
\end{algorithmic}
\label{algo:synthesis}
\end{algorithm}

The following theorem establishes correctness of this approach.
\begin{theorem}
\label{theorem:alternating-descent-CBF}
If the procedure described above terminates at stage $k$ with $\rho_{k} \leq 0$, then $b^{k-1}(x)$ is a CBF and $\{b^{k-1}(x) \geq 0\}$ is controlled positive invariant. If the procedure terminates at stage $k$ with $\rho_{k}^{\prime} \leq 0$, then $b^{k}(x)$ is a CBF and $\{b^{k}(x) \geq 0\}$ is controlled positive invariant. The procedure converges in a finite number of iterations.
\end{theorem}
\begin{proof}
We will show that if $\rho_{k} \leq 0$, then $b^{k-1}(x)$ is a CBF and $\{b^{k-1}(x) \geq 0\}$ is controlled positive invariant. If $\rho_{k} \leq 0$, then by construction of (\ref{eq:alternating-descent-1}) there exist SOS polynomials $\bar{\alpha}(x)$ and $\bar{\beta}(x)$ such that
\begin{IEEEeqnarray}{rCl}
\nonumber
\overline{\alpha}(x) &=& \alpha^{k}(x)\frac{\partial b^{k-1}}{\partial x}f(x) + \sum_{i=1}^{m}{\theta_{i}^{k}(x)\left[\frac{\partial b^{k-1}}{\partial x}g(x)\right]_{i}} \\
\label{eq:alternating-descent-proof-1}
&& +\eta^{k}(x)b^{k-1}(x) + \rho_{i}\Lambda(x) - 1\\
\label{eq:alternating-descent-proof-2}
\overline{\beta}(x) &=& -\beta^{k}(x)b^{k-1}(x) + w(x)h(x) - 1
\end{IEEEeqnarray}
Rearranging (\ref{eq:alternating-descent-proof-1}) and setting $\overline{\theta} = -\theta_{i}^{k}$ and $\overline{\eta} = -\eta^{k}$ yields
\begin{multline*}
\sum_{i=1}^{m}{\overline{\theta}_{i}(x)\left[\frac{\partial b^{k-1}}{\partial x}g(x)\right]_{i}} + \overline{\eta}(x)b^{k-1}(x) \\
+ \overline{\alpha}(x) - \rho_{k}\Lambda(x) - \alpha^{k}(x)\frac{\partial b^{k-1}}{\partial x}f(x) + 1 = 0
\end{multline*}
Since $\rho_{k} \leq 0$ and $\Lambda(x)$ is SOS, $\overline{\alpha}(x) - \rho_{k}\Lambda(x)$ is SOS. Hence,  there does not exist any $x$ satisfying $$\frac{\partial b^{k-1}}{\partial x}f(x) \leq 0, \ b^{k-1}(x) = 0, \ \frac{\partial b^{k-1}}{\partial x}g(x) = 0$$ and $b^{k-1}(x)$ is a CBF by Proposition \ref{prop:unconstrained-single-poly-viability}. The case where $\rho_{k}^{\prime} \leq 0$ is similar.

To prove convergence, we observe that $\rho_{k-1}^{\prime}$, $\alpha^{k-1}(x)$, $\eta^{k-1}(x)$, $\theta_{1}^{k-1}(x),\ldots,\theta_{m}^{k-1}(x)$, and $\beta^{k-1}(x)$ comprise a feasible solution to (\ref{eq:alternating-descent-1}) and $\rho_{k}$, $b^{k-1}(x)$ are feasible solutions to (\ref{eq:alternating-descent-2}). Hence $\rho_{1} \geq \rho_{1}^{\prime} \geq \cdots \geq \rho_{k} \geq \rho_{k}^{\prime} \geq \cdots$, i.e., the sequence is monotone nonincreasing. If the sequence is unbounded, then there exists $K$ such that $\rho_{k}, \rho_{k}^{\prime} \leq 0$ for all $k > K$, and hence the algorithm terminates after $K$ iterations. If the sequence is bounded below by zero, then it converges by the monotone convergence theorem and hence $|\rho_{k}-\rho_{k}^{\prime}| < \epsilon$ and $|\rho_{k} - \rho_{k-1}^{\prime}| < \epsilon$ for $k$ sufficiently large.
\end{proof}

We next consider synthesis of HOCBFs. We initialize functions $b^{0}(x),\psi_{0}^{0}(x),\ldots,\psi_{r}^{0}(x)$ arbitrarily. We initialize parameters $\rho_{0}$ and $\rho_{0}^{\prime}$ to be infinite. We then solve an SOS problem with variables $\rho$, $\eta^{k}(x)$, $\theta_{1}^{k}(x),\ldots,\theta_{m}^{k}(x)$, $\{\lambda_{S}(x) : S \subseteq \{0,\ldots,r+1\}\}$, $\{\beta_{S}^{k}(x) : S \subseteq \{0,\ldots,r\}\}$, and $w^{k}(x)$ under two constraints. The first constraint ensures viability and is given by
\begin{multline}
\label{eq:alternating-HOCBF-constraint-1}
\rho\Lambda(x) + \eta^{k}(x)\psi_{r}^{k-1}(x) + \sum_{i=1}^{m}{\theta_{i}^{k}(x)L_{f}^{r-1}b^{k-1}(x)} \\
 + \sum_{S  \subseteq \{0,\ldots,(r+1)\}}{\left[\lambda_{S}(x)(L_{f}^{r}b(x) + O(b(x)))\right.} \\
 \cdot \left.\prod_{l \in S}{\psi_{l}^{k}(x)}\right] 
  + (L_{f}^{r}b(x) + O(b(x)))^{2s} \in SOS
  \end{multline}
The second constraint ensures containment within the safe region and is given by
\begin{multline}
\label{eq:alternating-HOCBF-constraint-2}
\sum_{S  \subseteq \{0,\ldots,r\}}{\beta_{S}^{k}(x)\prod_{l \in S}{\psi_{l}^{k-1}(x)}} + w^{k}(x)h(x) + 1 \in SOS
\end{multline}
Hence the overall SOS program can be stated as
\begin{equation}
\label{eq:alternating-descent-HOCBF-1}
\begin{array}{ll}
\mbox{minimize} & \rho \\
\mbox{s.t.} &  (\ref{eq:alternating-HOCBF-constraint-1}), (\ref{eq:alternating-HOCBF-constraint-2}) \\
 & \lambda_{S}^{k}(x), \beta_{S}^{k}(x) \in SOS
 \end{array}
 \end{equation}

We let $\rho_{k}$ denote the value of $\rho$ returned by solving (\ref{eq:alternating-descent-HOCBF-1}). The procedure terminates if $\rho_{k} \leq 0$ or if $|\rho_{k}-\rho_{k-1}^{\prime}| < \epsilon$. Otherwise, the algorithm proceeds by solving the SOS problem with variables $\rho$, $b^{k}$, and $\psi_{0}^{k},\ldots,\psi_{r}^{k}$ and constraints
\begin{multline}
\label{eq:alternating-HOCBF-constraint-3}
\rho\Lambda(x) + \eta^{k}\psi_{r}^{k}(x) + \sum_{i=1}^{m}{\theta_{i}^{k}(x)L_{f}^{r-1}b^{k}(x)} \\
  + \sum_{S  \subseteq \{0,\ldots,(r+1)\}}{\lambda_{S}(x)(L_{f}^{r}b^{k}(x) + O(b^{k}(x)))}\\
  \cdot \prod_{l \in S}{\psi_{l}^{k}(x)} \in SOS
  \end{multline}
  and
  \begin{multline}
  \label{eq:alternating-HOCBF-constraint-4}
\sum_{S  \subseteq \{0,\ldots,r\}}{\beta_{S}^{k}(x)\prod_{l \in S}{\psi_{l}^{k-1}(x)}} + w^{k}(x)h(x) + 1 \in SOS
\end{multline}
The SOS program is stated as
 \begin{equation}
 \label{eq:alternating-descent-HOCBF-2}
 \begin{array}{ll}
 \mbox{minimize} & \rho \\
 \mbox{s.t.} &  (\ref{eq:alternating-HOCBF-constraint-3}), (\ref{eq:alternating-HOCBF-constraint-4})
 \end{array}
 \end{equation}
 The procedure terminates if $\rho^{\prime} \leq 0$ or if $|\rho_{k}-\rho_{k}^{\prime}| < \epsilon$. Similar to Theorem \ref{theorem:alternating-descent-CBF}, it can be shown that this approach terminates within finite time, and returns an HOCBF if $\rho \leq 0$ or $\rho^{\prime} \leq 0$ at convergence.
 
 \subsection{CBF Computation at Equilibria}
 \label{subsec:fixed-points}
 We next describe how structures such as the existence of fixed points  can be leveraged to construct CBFs. We  consider the case where the system $\dot{x}(t) = f(x) + g(x)u$ has a fixed point $(x^{\ast},u^{\ast})$ satisfying $f(x^{\ast}) + g(x^{\ast})u^{\ast} = 0$, and have the following result.
 
 \begin{lemma}
 \label{lemma:fixed-point}
 Suppose that (\ref{eq:dynamics}) has $(x^{\ast},u^{\ast})$ as a fixed point and the linearization around $(x^{\ast},u^{\ast})$ is controllable. Then there exists $\epsilon > 0$ and a positive definite matrix $P$ such that $b(x) = \epsilon - (x-x^{\ast})^{T}P(x-x^{\ast})$ is a CBF.
 \end{lemma}
 
 \begin{proof}
 Let $\dot{x}(t) = Fx(t) + Gu(t)$ be the linearization of (\ref{eq:dynamics}) in a neighborhood of $(x^{\ast},u^{\ast})$. Let $\tilde{u}(t) = u^{\ast}-K(x-x^{\ast})$ be an exponentially stabilizing controller of the linearized system. Then there is a quadratic Lyapunov function $V(x) = (x-x^{\ast})^{T}P(x-x^{\ast})$ such that (\ref{eq:dynamics}) is asymptotically stable in a neighborhood of $x^{\ast}$, i.e., $$\frac{\partial V}{\partial x}(f(x) + g(x)\tilde{u}(t)) \leq 0.$$ This is equivalent to $\frac{\partial b}{\partial x}(f(x) + g(x)\tilde{u}) \geq 0$ for $(x-x^{\ast})$ sufficiently small. 
 \end{proof}
 
 The preceding lemma implies that we can construct a CBF by computing a fixed point $(x^{\ast},u^{\ast})$ and a stabilizing linear controller $\overline{u} = -K\overline{x}$ of the linearized system and solving the Lyapunov equation $\overline{F}^{T}P + P\overline{F} + N = 0$ where $\overline{F} = (F-GK)$ and $N$ is a positive definite matrix. It remains to choose the parameter $\epsilon$, which can be performed via the following procedure.
 
 \begin{theorem}
 \label{theorem:fixed-point}
 Suppose that $(x^{\ast},u^{\ast})$ is a fixed point of (\ref{eq:dynamics}), $P$ is a solution to the Lyapunov equation for a stabilizing controller of the linearized sytsem at $(x^{\ast},u^{\ast})$, and the SOS constraints
 \begin{multline}
 \label{eq:fixed-point-1}
 \alpha(x)(-2(x-x^{\ast})^{T}Pf(x)) + \sum_{i=1}^{m}{\theta_{i}(x)[-2(x-x^{\ast})^{T}Pg(x)]_{i}} \\
 + \eta(x)(\epsilon - (x-x^{\ast})^{T}P(x-x^{\ast})) \in SOS
 \end{multline}
 and
 \begin{equation}
 \label{eq:fixed-point-2}
 h(x)-\beta(x)(\epsilon - (x-x^{\ast})^{T}P(x-x^{\ast})) \in SOS
 \end{equation}
 hold for some polynomials $\theta_{1},\ldots,\theta_{m}$, $\eta(x)$ and SOS polynomials $\alpha(x)$ and $\beta(x)$. Then $b(x) = \epsilon - (x-x^{\ast})^{T}P(x-x^{\ast}))$ is a CBF and $\{b(x) \geq 0\}$ is in the safe region.
 \end{theorem}
 
 \begin{proof}
 If (\ref{eq:fixed-point-1}) holds, then $b$ is a CBF by Propostion \ref{prop:unconstrained-single-poly-viability}. By nonnegativity of SOS polynomials, we have that $b(x) \geq 0$ implies $h(x) \geq 0$ if $h(x) = \beta(x)b(x) + \beta_{0}(x)$ for SOS polynomials $\beta$ and $\beta_{0}$. Rearranging terms gives (\ref{eq:fixed-point-2}).
 \end{proof}

\section{Simulation Study}
\label{sec:simulation}
We conducted two simulation studies to evaluate our proposed approach. The first study considered the problem of verifying safety of a quadrotor UAV. The second study considered the problem of synthesizing a CBF for a power converter. Finally, we considered equilibrium-based CBF synthesis for a linearized quadrotor with actuation constraints.

\subsection{Case Study: Quadrotor UAV Safety Verification}
\label{subsec:sim-verification}
We consider a two-dimensional quadrotor model. In what follows, we describe the system dynamics, safety constraint, and verification results.

\subsubsection{Quadrotor Dynamics}
The quadrotor has a time-varying state $x(t) \in \mathbb{R}^{6}$. The states $x_{1}$ and $x_{2}$ represent horizontal and vertical position, $x_{3}$ represents yaw angle, and $x_{4},x_{5},x_{6}$ represent horizontal, vertical, and yaw velocities, respectively. The dynamics are given by
\begin{equation}
    \label{eq:quadrotor-dynamics}
    \left(
    \begin{array}{c}
    \dot{x}_{1} \\
    \dot{x}_{2} \\
    \dot{x}_{3} \\
    \dot{x}_{4} \\
    \dot{x}_{5} \\
    \dot{x}_{6}
    \end{array}
    \right) = \left(
    \begin{array}{c}
    x_{4} \\
    x_{5} \\
    x_{6} \\
    0 \\
    -g \\
    0
    \end{array}
    \right) + \left(
    \begin{array}{cc}
    0 & 0 \\
    0 & 0 \\
    0 & 0 \\
    -\frac{\sin{x_{3}}}{M_{q}} & -\frac{\sin{x_{3}}}{M_{q}} \\
    \frac{\cos{x_{3}}}{M_{q}} & \frac{\cos{x_{3}}}{M_{q}} \\
    \frac{L_{r}}{I_{n}} & \frac{L_{r}}{I_{n}}
    \end{array}
    \right)\left(
    \begin{array}{c}
    u_{1} \\
    u_{2}
    \end{array}
    \right)
\end{equation}
where $g=9.8$ is the gravitational constant, $M_{q} = 0.486$ is the mass of the quadrotor, $I_{n} = 0.00383$ is the moment of inertia, and $L_{r} = 0.25$ is the length of the rotor arm.

\subsubsection{Safety Constraints} We consider the problem of synthesizing an invariant set in order to enforce a constraint on the minimum altitude of five meters for  the quadrotor, so that the safe region is defined by $\{x : x_{2} \geq 5\}$. This set is clearly not controlled positive invariant, since $\dot{x}_{2} = x_{5}$, which does not contain the control input $u$. In order to construct a safety constraint, we first use the approach of Section \ref{subsec:trig-functions} to convert \eqref{eq:quadrotor-dynamics} to an equivalent polynomial system
\begin{equation}
    \label{eq:quadrotor-polynomial}
    \left(
    \begin{array}{c}
    \dot{\overline{x}}_{1} \\
    \dot{\overline{x}}_{2}\\
    \dot{\overline{x}}_{3}\\
    \dot{\overline{x}}_{4}\\
    \dot{\overline{x}}_{5}\\
    \dot{\overline{x}}_{6}\\
    \dot{\overline{x}}_{7}
    \end{array}
    \right) = \left(
    \begin{array}{c}
    \overline{x}_{4} \\
    \overline{x}_{5} \\
    \overline{x}_{6}\overline{x}_{7} \\
    0 \\
    -g \\
    0 \\
    -\overline{x}_{3}\overline{x}_{6}
    \end{array}
    \right) + \left(
    \begin{array}{cc}
    0 & 0 \\
    0 & 0 \\
    0 & 0 \\
    -\frac{x_{3}}{M_{q}} & -\frac{x_{3}}{M_{q}} \\
    \frac{x_{7}}{M_{q}} & \frac{x_{7}}{M_{q}} \\
    \frac{L_{r}}{I_{n}} & -\frac{L_{r}}{I_{n}} \\
    0 & 0
    \end{array}
    \right)\left(
    \begin{array}{c}
    u_{1} \\
    u_{2}
    \end{array}
    \right)
\end{equation}
where $\overline{x}_{3} = \sin{x_{3}}$ and $\overline{x}_{7} = \cos{x_{3}}$. In order to ensure that $\bar{x}_{2} \geq 5$, we include an additional safety constraint $\bar{x}_{5} - k(\bar{x}_{2}-5) \geq 0$ for some $k \geq 0$. This constraint ensures that, when $x_{2} = 5$ (the quadrotor reaches the boundary of the safe region), we have $\dot{x}_{2} \geq 0$. Examining $\dot{\bar{x}}_{5}$, we find that $L_{g}\overline{x}_{5} = 0$ when $\bar{x}_{7} = 0$. We therefore add a third safety constraint $\{\bar{x}_{7} \geq \delta\}$ for some $\delta > 0$. Finally, to ensure that the constraint on $\bar{x}_{7}$ is satisfied, we add the constraint $\{-\bar{x}_{3}\bar{x}_{6} - \phi (x_{7}-\delta) \geq 0\}$. Introducing symmetric constraints when $\delta < 0$ gives the semi-algebraic invariant set candidate 
\begin{multline}
\label{eq:quadrotor-invariant}
    \mathcal{C} = \left(\{\bar{x}_{2} \geq 5\} \cap \{\bar{x}_{5} - k(\bar{x}_{2}-\beta) \geq 0\} \right.\\
    \left.\cap \{\bar{x}_{7} \geq \delta\} \cap \{-\bar{x}_{3}\bar{x}_{6}-\phi(\bar{x}_{7}-\delta) \geq 0\}\right) \\
    \cup\left(\{\bar{x}_{2} \geq 5\} \cap \{\bar{x}_{5} - k(\bar{x}_{2}-\beta) \geq 0\}\right. \\
    \left.\cap \{-\bar{x}_{7}- \delta \geq 0\} \cap \{\bar{x}_{3}\bar{x}_{6}-\phi(\bar{x}_{7}-\delta) \geq 0\}\right)
\end{multline}

\subsubsection{Verification Procedure} We can apply the verification procedure of Section \ref{subsec:viability-verification} to check that $\mathcal{C}$ is viable. We have can write $\mathcal{C}$ in the form \eqref{eq:closed-semialgebraic} with $s=2$, $r_{1}=r_{2}=4$, 
\begin{IEEEeqnarray}{rCl}
b_{11}(x) &=& \bar{x}_{2}-5, \quad b_{12}(x) = \overline{x}_{5}-k(\overline{x}_{2}-5) \\
b_{13}(x) &=& \bar{x}_{7}-\delta, \quad b_{14}(x) = -\bar{x}_{3}\bar{x}_{6}-\phi(\bar{x}_{7}-\delta) \\
b_{21}(x) &=& \bar{x}_{2}-5, \quad b_{22}(x) = \overline{x}_{5}-k(\overline{x}_{2}-5) \\
b_{23}(x) &=& -\bar{x}_{7}-\delta, \quad b_{24}(x) = \bar{x}_{3}\bar{x}_{6}-\phi(\bar{x}_{7}-\delta)
\end{IEEEeqnarray}
The possible values of $S$ are $\{1\}$, $\{2\}$, and $\{1,2\}$. For $S = \{1,2\}$, the set $\mathcal{C}_{1} \cap \mathcal{C}_{2}$ is empty, making verification trivial. When $S = \{1\}$, the set $\Pi(\{1,2\} \setminus \{1\}) = \{(2,1), (2,2), (2,3), (2,4)\}$ and $Z(S) = 2^{\{1,2,3,4\}}$ (that is, the set of subsets of $\{1,2,3,4\}$). The case where $S = \{2\}$ is similar. Altogether, we have 64 SOS programs to solve for each $S$, for a total of 128 SOS programs. However, we observe that some programs can be removed immediately as they involve verifying empty sets. Specifically, since $b_{11}(x) = b_{21}(x)$ and $b_{12}(x) = b_{22}(x)$, the cases where $\pi(2) = 1$ and $\pi(2) = 2$ (for $S =\{1\}$) and $\pi(1)=1$ and $\pi(1) = 2$ (for $S = \{1\}$) can be ruled out. With this reduction, there are 32 SOS programs to be solved.

The verification problem was solved on a MacBook M1 Pro with 16GB of memory using the Yalmip SOS solver \cite{lofberg2004yalmip}. The total runtime was 3260 seconds.

\subsection{Case Study: Power Converter Safety Synthesis}
\label{subsec:sim-synthesis}
We consider a power converter with state $x(t) \in \mathbb{R}^{3}$ and dynamics given by \cite{schneeberger2023sos}
\begin{multline}
    \label{eq:power-converter-dynamics}
    \left(
    \begin{array}{c}
    \dot{x}_{1} \\
    \dot{x}_{2} \\
    \dot{x}_{3} 
    \end{array}
    \right) = \left(
    \begin{array}{c}
    -0.05x_{1} - 57.9x_{2} + 0.00919x_{3} \\
    1710x_{1} + 314x_{3} \\
    -0.271x_{1} - 314x_{2}
    \end{array}
    \right) \\
    + \left(
    \begin{array}{cc}
    0.05-57.9x_{2} & -57.9x_{3} \\
    1710+1710x_{1} & 0 \\
    0 & 1710 + 1710x_{1}
    \end{array}
    \right)\left(
    \begin{array}{c}
    u_{1} \\
    u_{2}
    \end{array}
    \right)
\end{multline}
The safety constraint is given by $\{(x_{1}/20)^{2} + x_{2}^{2} + x_{3}^{2} - 1.2^2 \leq 0\}$. The goal of this case study is to evaluate the basin of attraction of the  alternating-descent synthesis approach of Section \ref{subsec:ad-synthesis}. Specifically, we construct a CBF $b$ that satisfies the positive invariance conditions with $\{b(x) \geq 0\}$ contained in the safe region. We then randomly perturb $b(x)$ to obtain a new function $b_{0}(x)$ and initialize the alternating descent algorithm with starting point $b_{0}(x)$. The level of perturbation of $b(x)$ that can be tolerated while still returning a valid CBF will be used to evaluate the CBF synthesis algorithm.

In order to construct a CBF $b$, we linearize \eqref{eq:power-converter-dynamics} around the equilibrium point $(0,0,0)$. We then use the equilibrium-based synthesis of Section \ref{subsec:fixed-points} to compute a quadratic CBF. The CBF is given by $b(x) = c - x^{T}Sx$, where 
\begin{IEEEeqnarray*}{rCl}
    S &=& \left(
    \begin{array}{ccc}
    0.0251 & -0.000239 & -7.30 \times 10^{-5} \\
    -0.000239 & 0.000600 & -5.17 \times 10^{-5} \\
    -7.30 \times 10^{-5} & -5.17 \times 10^{-5}
    \end{array}
    \right)\\
     c &=& 7.587 \times 10^{-4}
\end{IEEEeqnarray*}

We then perturb this condition by setting $b_{0}(x) = c - x^{T}S_{0}x$, where $S_{0} = (I + \Delta )\odot S$, $I$ is the 3-by-3 identity matrix, $\odot$ denotes the entrywise matrix product, and $\Delta$ is a matrix with i.i.d. Gaussian entries that have mean zero and variance $\sigma$. The parameter $\sigma$ was varied in order to test the effect of different perturbation levels. For each $\sigma \in \{1,2,3,\cdots,20\}$, the alternating-descent algorithm was repeated 50 times. The fraction of the time that the algorithm terminated with a valid CBF is given in Table \ref{table:alternating-results}. 
\begin{table}[h]
    \caption{Success probability of alternating synthesis}
    \label{table:alternating-results}
    \begin{tabular}{|c|c|c|c|c|c|c|c|}
    \hline
    $\sigma$ & 1 & 4 & 7 & 10 & 13 & 16 & 20 \\
    \hline
    Success prob. & 0.72 & 0.6 & 0.72 & 0.6 & 0.62 & 0.78 & 0.74 \\
    \hline
    \end{tabular}
\end{table}
We observe that the algorithm converges even for large perturbations of the initial condition. One possible reason for this behavior is that the system is open-loop stable, which reduces the CBF synthesis problem to synthesis of a Lyapunov function. In order to evaluate our synthesis algorithm on an open-loop unstable system, we perturbed the dynamics to obtain
\begin{multline}
    \label{eq:power-converter-dynamics-2}
    \left(
    \begin{array}{c}
    \dot{x}_{1} \\
    \dot{x}_{2} \\
    \dot{x}_{3} 
    \end{array}
    \right) = \left(
    \begin{array}{c}
    -0.057x_{1} - 50.9x_{2} + 0.0096x_{3} \\
    1952x_{1} + 314x_{3} \\
    -0.2982x_{1} - 365.2x_{2}
    \end{array}
    \right) \\
    + \left(
    \begin{array}{cc}
    0.05-57.9x_{2} & -57.9x_{3} \\
    1710+1710x_{1} & 0 \\
    0 & 1710 + 1710x_{1}
    \end{array}
    \right)\left(
    \begin{array}{c}
    u_{1} \\
    u_{2}
    \end{array}
    \right)
\end{multline}
We followed the same synthesis procedure and found that the the proposed alternating descent approach converged to a valid CBF 70\% of the time with $\sigma=1$. 

\subsection{Synthesis for Linearized Quadrotor}

We simulated our approach on a quadrotor UAV with state $x(t) \in \mathbb{R}^{n}$ and control input $u(t) \in \mathbb{R}^{m}$. We have
%\begin{equation}
%    \label{eq:quadrotor-dynamics}
%    f(x) = \left(
%    \begin{array}{c}
%    x_{4} \\
%    x_{5} \\
%    x_{6} \\
%    0 \\
%    -g \\
%    0
%    \end{array}
%    \right), \quad g(x) = \left(
%    \begin{array}{cc}
%    0 & 0 \\
%    0 & 0 \\
%    0 & 0 \\
%    -\frac{\sin{x_{2}}}{M_{q}} & -\frac{\sin{x_{2}}}{M_{q}} \\
%    \frac{\cos{x_{2}}}{M_{q}} & \frac{\cos{x_{2}}}{M_{q}} \\
%    \frac{L_{r}}{I_{n}} & -\frac{L_{r}}{I_{n}}
%    \end{array}
%    \right)
%\end{equation}
%In the above $g = 9.8$ is the gravitational constant, $M_{q}=0.486$kg is the mass of the quadrotor, $L_{r}=0.25$m is the length of the rotor, and $I_{n}=0.00383$. We assume that the control input $u(t)$ satisfies the constraint $u(t) \in \mathcal{U} = \{u : ||u||_{\infty} \leq c\}$. We repeated the simulation for $c \in \{0,1,1,10\}$ to evaluate the change in the safe region. The outer safety constraint was given by $\mathcal{S} = \{x : ||x||_{2}^{2} \leq 1\}$. 

In order to construct a CBF, we first linearized (\ref{eq:quadrotor-dynamics}) around the origin to obtain the linearized model $\dot{x}(t) = Fx + Gu$ and constructed an LQR controller. We defined $b(x) = \epsilon - x^{T}Px$, where $P$ is the solution to the Lyapunov equation for the LQR controller, as in Section \ref{subsec:fixed-points}. 
We found the maximum value of $\epsilon$ such that viability of $\mathcal{C}$ was satisfied and $\mathcal{C} \subseteq \mathcal{S}$ by solving an SOS program as in Section \ref{subsec:fixed-points}. %\textbf{*Elaborate on how SOSTOOLS was used.}

Fig. \ref{fig:simulation}(a) shows the shape of the safe region for different sizes of the constraint set, confirming that the safe region contracts as the control constraints become more restrictive. 

Figure \ref{fig:simulation}(b) shows the trajectory of the UAV for different initial states. We choose $c = 1$ and use a CBF-based controller with $b(x)$ as defined above. The CBF-based controller selects $u$ at each time step according to
\begin{equation}
    \label{eq:quadrotor-control-law}
    \begin{array}{ll}
\mbox{minimize} & u^{T}u \\
\mbox{s.t.} & \frac{\partial b}{\partial x}(f(x)+g(x)u) \geq -kb(x)\\
 & ||u|| \leq 1
\end{array}
\end{equation}
with $k=10$. As shown in the figure, if $x(0) \notin \mathcal{C}$ (red curve), the trajectory eventually exits the safe region even if the CBF control policy (\ref{eq:quadrotor-control-law}) is used. On the other hand, if $x(0) \in \mathcal{C}$, then $x(t) \in \mathcal{C} \subseteq \mathcal{S}$ for all time $t \geq 0$.

\begin{figure*}[!ht]
\centering
$\begin{array}{cc}
\includegraphics[width=3in]{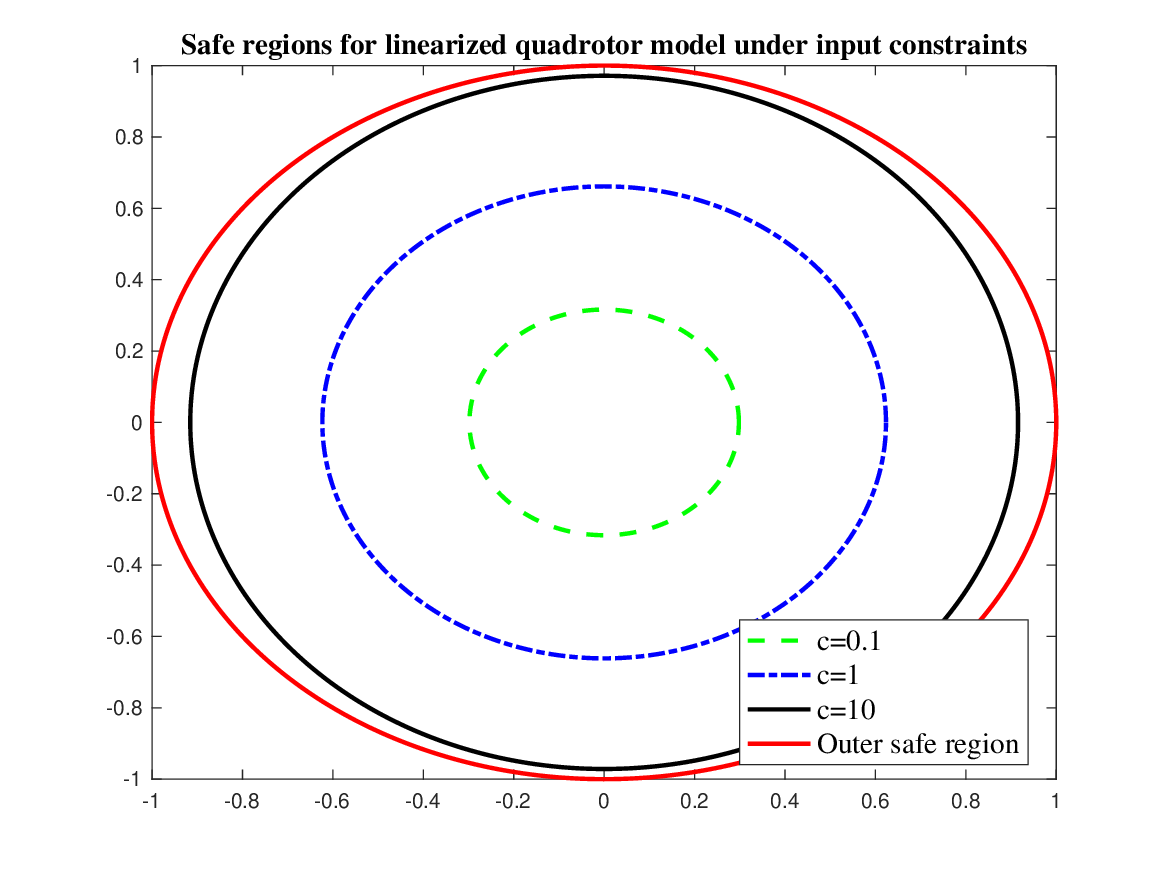} & 
\includegraphics[width=3in]{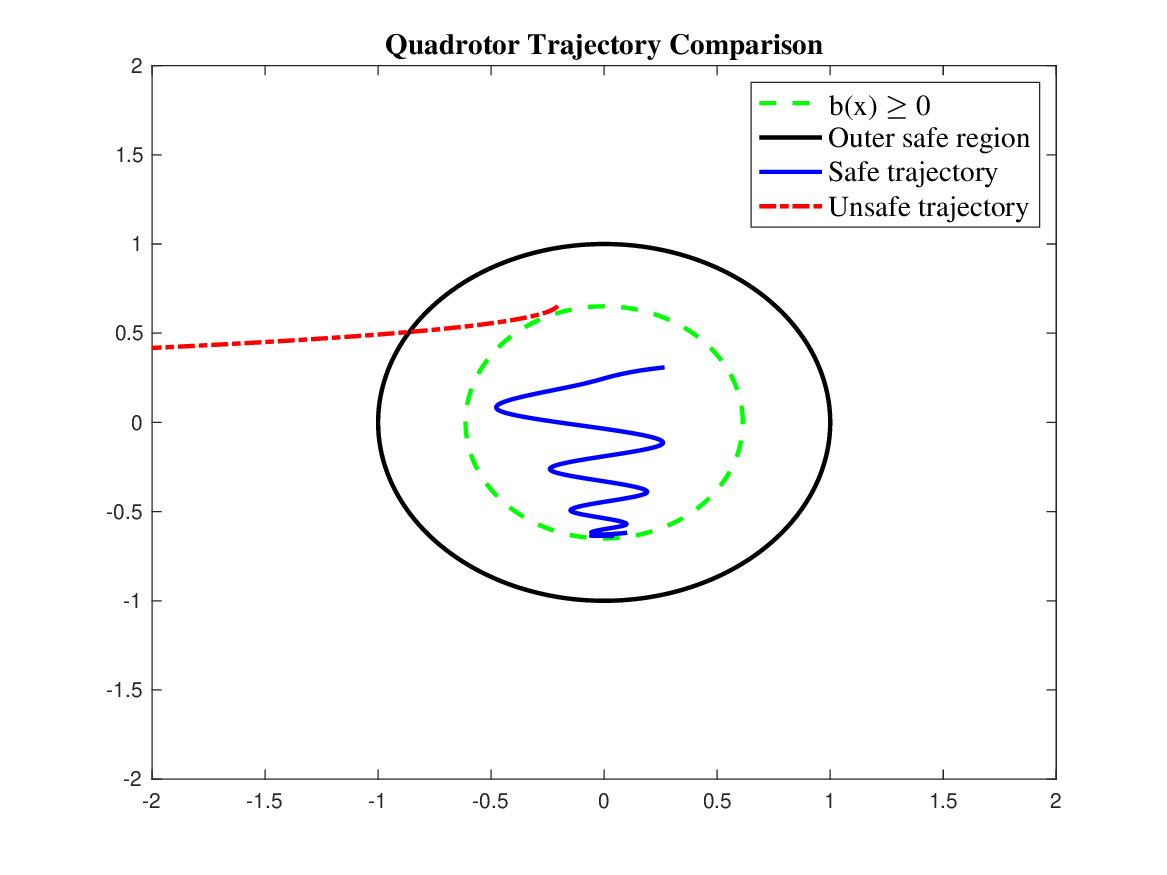} \\
\mbox{(a)} & \mbox{(b)}
\end{array}$
\caption{Numerical evaluation of our approach on a linearized quadrotor model. (a) Geometry of the safe region for different values of the input constraint. As the input constraint becomes more restrictive, the size of the feasible safe region is reduced. (b) Comparison of trajectories from CBF-based control policies. In both cases, the control policy is given by (\ref{eq:quadrotor-control-law}) with $k=10$. Initializing the state within $\{x: b(x) \geq 0\}$ maintains the state within the safe region $\{x: ||x||_{2}^{2} \leq 1\}$, while initializing outside $\mathcal{C}$ results in a safety violation.}
\label{fig:simulation}
\end{figure*}
\section{Conclusions}
\label{sec:conclusion}
This paper studied verification and synthesis of control barrier functions and controlled positive invariant sets. We showed that verifying that a semi-algebraic set is controlled positive invariant for a system with polynomial dynamics is equivalent to the non-existence of a solution to a system of polynomial equations and inequalities. We then proved that controlled positive invariance is equivalent to a sum-of-squares optimization problem via the Positivstellensatz. We derived techniques for safety verification of CBF based controllers based on our framework, as well as sufficient conditions for existence of continuous safe feedback control policies.

We proposed two algorithms for synthesizing control barrier functions. The first algorithm is an alternating descent heuristic based on solving a sequence of sum-of-squares programs. The second algorithm derives a local CBF in the neighborhood of a fixed point of the system. We evaluated our approach through  simulation studies on a  quadrotor model and a power converter system. %The third algorithm uses Newton's method on function space. We proved that the latter algorithm has guaranteed local quadratic convergence to a CBF.

\bibliographystyle{IEEEtran}
\bibliography{TAC-CBF}

\appendices
\section{Verifying Practicality of a Semi-Algebraic Set}
\label{appendix:practical}
In what follows, we describe an approach to verifying that a semi-algebraic set is practical. We have the following result.

\begin{theorem}
    \label{prop:pratical-verification}
  Let $$\mathcal{C} = \cap_{j=1}^{r}{\{x: b_{i}(x) \geq 0\}}$$ be a closed simple semi-algebraic set and let $\epsilon > 0$. For any $S \subseteq \{1,\ldots,r\}$, define $\Theta_{S}(x)$ and $\Phi_{S}(x)$ to be the matrices with rows equal to $(\frac{\partial b_{i}}{\partial x} : i \in S)$ and $(\frac{\partial b_{i}}{\partial x} : i \notin S)$ respectively. The set $\mathcal{C}$ is practical if and only if, for all $S \subseteq \{1,\ldots,r\}$, there do not exist $x \in \mathbb{R}^{n}$, $y \in \mathbb{R}^{|S|}$, and $w \in \mathbb{R}^{r-|S|}$ satisfying (i) $b_{i}(x) = 0$ for all $i \in S$, (ii) $b_{i}(x) > 0$ for all $i \in \{1,\ldots,r\} \setminus S$, (iii) $y^{T}\Theta_{S}(x) + w^{T}\Phi_{S}(x) = 0$, (iv) $y,w\geq 0$ and (v) $y^{T}y \geq \epsilon^{2}$.
\end{theorem}

\begin{proof}
Recall that a semi-algebraic set is practical if, for any $x \in \mathcal{C}$, there exists $z$ satisfying $\frac{\partial b_{i}}{\partial x}z + b_{i}(x) > 0$. Suppose that $\mathcal{C}$ is practical and  $x \in \mathcal{C}$, and let $S = \{i : b_{i}(x) = 0\}$. Since $x \in \mathcal{C}$, we must have $b_{i}(x) > 0$ for $i \notin S$. Hence the condition for practicality implies that there exists $z$ satisfying $\frac{\partial b_{i}}{\partial x}z > 0$ for all $i \in S$ and $\frac{\partial b_{i}}{\partial x}z > -b_{i}(x)$ for all $i \notin S$. By scaling $z$ appropriately, we have that this is equivalent to the existence of $z$ satisfying $\frac{\partial b_{i}}{\partial x}z > 0$ for $i \in S$ and $\frac{\partial b_{i}}{\partial x}z \geq 0$ for $i \notin S$. We can write these inequalities in matrix form as $\Theta_{S}(x)z > 0$ and $\Phi_{S}x \geq 0$. By Motzkin's Transposition Theorem, the existence of such a $z$ is equivalent to the non-existence of $y$ and $w$ satisfying $y^{T}\Theta_{S}(x) + w^{T}\Phi_{S}(x) = 0$, $y,w \geq 0$, and $y \neq 0$. By scaling $y$, this can be made equivalent to $y^{T}y \geq \epsilon^{2}$. Hence, if $x \in \mathcal{C}$, then conditions (i)-(v) are satisfied.

Conversely, suppose that $\mathcal{C}$ is not practical. By the preceding analysis, there must exist $x$ and $z$ violating the conditions (i)-(v).
\end{proof}

The conditions (i)-(v) define a semi-algebraic set, and hence equivalent conditions for non-existence of $(x,y,w)$ satisfying (i)-(v) can be derived using the Positivstellensatz. In the case where $r=1$, the condition to be verified reduces to non-existence of $x$ satisfying $b(x) = 0$ and $\frac{\partial b}{\partial x} = 0$. By the Positivstellensatz, this is equivalent to the existence of polynomials $\eta(x)$ and $\theta_{1}(x),\ldots,\theta_{n}(x)$ satisfying $$\left(\eta(x)b(x) + \sum_{i=1}^{n}{\frac{\partial b}{\partial x_{i}}\theta_{i}(x)} -1\right) \in SOS.$$

\section{Proofs from Section \ref{subsec:viability-verification}}
\label{appendix:proofs}

\begin{proof}[Proof of Theorem \ref{prop:semialgebraic-boundary}]
Let $\mathcal{C}_{i} = \{x: b_{ij}(x) \geq 0 \ \forall j=1,\ldots,r_{i}\}$. Suppose that $x \in \partial \mathcal{C}$. Since $\mathcal{C}$ is closed, we have $x \in \mathcal{C}$, and hence there exists $S \subseteq \{1,\ldots,s\}$ such that $x \in \mathcal{C}_{i}$ for $i \in S$ and $x \notin \mathcal{C}_{i}$ for $i \notin S$. If $i \notin S$, there must exist at least one $j_{i} \in \{1,\ldots,r_{i}\}$ such that $b_{ij_{i}}(x) < 0$. If $i \in S$, then we must have $b_{ij}(x) \geq 0$ for all $j=1,\ldots,r_{i}$. Furthermore, there must exist at least one index $j$ such that $b_{ij}(x) = 0$, since otherwise $x$ would lie in the interior of $\mathcal{C}_{i}$ and hence in the interior of $\mathcal{C}$. Letting $T_{i}$ denote the collection of such indices completes  the proof.
\end{proof}

\begin{proof}[Proof of Lemma \ref{lemma:tangent-cone-semi-algebraic}]
The first step in the proof is to show that $\mathcal{T}_{\mathcal{C}}(x) = \bigcup_{i \in S}{\mathcal{T}_{\mathcal{C}_{i}}(x)}$. Suppose that $z \in \mathcal{T}_{\mathcal{C}_{i}}(x)$. Then $$\liminf_{\tau \rightarrow 0}{\frac{\mbox{dist}(x+\tau z, \mathcal{C})}{\tau}} \leq \liminf_{\tau \rightarrow 0}{\frac{\mbox{dist}(x+\tau z,\mathcal{C}_{i})}{\tau}} = 0,$$ implying that $z \in \mathcal{T}_{\mathcal{C}}(x)$.

Now, let $z \in \mathcal{T}_{\mathcal{C}}(x)$, and suppose $z \notin \mathcal{T}_{\mathcal{C}_{i}}(x)$ for all $i \in S$. Then there exist $\epsilon_{i} > 0$ for $i \in S$ such that $$\liminf_{\tau \rightarrow 0}{\frac{\mbox{dist}(x+\tau z, \mathcal{C}_{i})}{\tau}} = \epsilon_{i}$$ Hence there exist $\delta > 0$ and $\overline{\tau} > 0$ such that $\tau < \overline{\tau}$ implies that $\frac{\mbox{dist}(x+\tau z, \mathcal{C}_{i})}{\tau} > \delta$ for all $i$, implying that $\frac{\mbox{dist}(x+\tau z, \mathcal{C})}{\tau} > \delta$ for all $\tau < \overline{\tau}$. This, however, contradicts the assumption that $z \in \mathcal{T}_{\mathcal{C}}(x)$, and hence we must have $z \in \mathcal{T}_{\mathcal{C}_{i}}(x)$ for some $i \in S$.

The fact that $$\mathcal{T}_{\mathcal{C}_{i}}(x) = \left\{z: \frac{\partial b_{ij}}{\partial x}z \geq 0\right\}$$ follows from the definition of $\mathcal{C}_{i}$ and Lemma \ref{lemma:tangent-cone-intersection}. Combining this with $\mathcal{T}_{\mathcal{C}}(x) = \bigcup_{i \in S}{\mathcal{T}_{\mathcal{C}_{i}}(x)}$ completes the proof. 
\end{proof}

\begin{proof}[Proof of Lemma \ref{lemma:Psatz-Farkas}]
By Farkas Lemma, there exists a solution $u$ to the system $\Theta(x)u \leq \psi(x)$ if and only if there is no $y \in \mathbb{R}^{N}$ satisfying $y_{i} \geq 0$ for $i=1,\ldots,N$, $\Theta(x)^{T}y = 0$, and $\psi(x)^{T}y < 0$. Hence, for all $x$ with $p(x) > 0$ $\forall p \in \mathcal{P}$ and $q(x) = 0$ $\forall q \in \mathcal{Q}$, there exists a solution to at least one of $\Theta_{i}(x)u \leq \psi_{i}(x)$  if and only if there do not exist $(x,y) \in \mathbb{R}^{n+N}$ such that $y_{1},\ldots,y_{N} \geq 0$, $\Theta_{i}(x)^{T}y = 0$ for $i=1,\ldots,m$, $-\psi(x)^{T}y < 0$, $p(x) > 0$ for $p \in \mathcal{P}$ and $q(x) = 0$ for $q \in \mathcal{Q}$. The non-existence of such points is equivalent to the existence of the polynomials $\lambda$, $\phi$, and $\sigma$ defined in the statement of the lemma by Theorem \ref{theorem:Psatz}. 
\end{proof}

\begin{proof}[Proof of Theorem \ref{theorem:semialgebraic-viability}]
We first prove that $\mathcal{C}$ is viable if the conditions of the theorem are met. Suppose $x \in \partial \mathcal{C}$. By Theorem \ref{prop:semialgebraic-boundary}, there exist $S \subseteq \{1,\ldots,s\}$, $\pi \in \Pi(\{1,\ldots, s\} \setminus S)$, and $(T_{i} : i \in S) \in Z(S)$ such that $b_{i\pi(i)}(x) < 0$ for $i \notin S$, $b_{ij}(x) > 0$ for $i \in S$, $j \notin T_{i}$, and $b_{ij}(x) = 0$ for $i \in S$ and $j \in T_{i}$. This is equivalent to $p(x) > 0$ for $p \in \mathcal{P}_{S,T,\pi}$ and $q(x) = 0$ for $q \in \mathcal{Q}_{S,T,\pi}$. 

By Lemma \ref{lemma:Psatz-Farkas}, $\mathcal{B}(\{\Theta_{S,T,l} : l \in S\}, \{\psi_{S,T,l} : l \in S\}, \mathcal{P}_{S,T,\pi}, \mathcal{Q}_{S,T,\pi})$ is nonempty if and only if for every $x$ with $p(x) > 0$ for $p \in \mathcal{P}$ and $q(x) = 0$ for $q \in \mathcal{Q}$, there exists $l \in S$ and a vector $u$ satisfying $\Theta_{S,T,l}(x)u \leq \psi_{S,T,l}(x)$. By definition of $\Theta_{S,T,l}$ and $\psi_{S,T,l}$, we then have (i) $\hat{\Theta}_{S,T,l}(x)u \leq \hat{\psi}_{S,T,l}(x)$ and (ii) $Au \leq c$. Constraint (i) implies that $$-\frac{\partial b_{lj}}{\partial x}g(x)u \leq \frac{\partial b_{lj}}{\partial x}f(x)$$ for all  $j \in T_{l}$, and hence $(f(x)+g(x)u) \in \mathcal{T}_{\mathcal{C}}(x)$ by Lemma \ref{lemma:tangent-cone-semi-algebraic}. Constraint (ii) implies that $u \in \mathcal{U}$. Hence, at every $x \in \partial \mathcal{C}$, there exists $u \in \mathcal{U}$ with $(f(x)+g(x)u) \in \mathcal{T}_{\mathcal{C}}(x)$, implying that $\mathcal{C}$ is viable.

Now, suppose that the conditions of the theorem do not hold for some $S \subseteq \{1,\ldots,s\}$, $\pi \in \Pi(\{1,\ldots,s\} \setminus S)$, and $T \in Z(S)$. By Lemma \ref{lemma:Psatz-Farkas}, there exists $x$ with $b_{i\pi(i)}(x) < 0$ for $i \notin S$, $b_{ij}(x) = 0$ for $i \in S$ and $j \in T_{i}$, and $b_{ij}(x) > 0$ for $i \in S$ and $j \notin T_{i}$, such that there for all $l$ there is no $u \in \mathcal{U}$ with $\hat{\Theta}_{S,T,l}(x)u \leq \hat{\psi}_{S,T,l}(x)$. Equivalently, there exists $x \in \Delta(\mathcal{C})$ such that there is no $u \in \mathcal{U}$ with $(f(x) + g(x)u) \in \mathcal{T}_{\mathcal{C}}(x)$. Finally, we observe that we must have $x \in \partial \mathcal{C}$, since if $x \notin \partial \mathcal{C}$ we would have $\mathcal{T}_{\mathcal{C}}(x) = \mathbb{R}^{n}$, a contradiction. Thus there exists $x \in \partial \mathcal{C}$ with $(f(x) + g(x)u) \notin \mathcal{T}_{\mathcal{C}}(x)$ for all $u \in \mathcal{U}$, implying that $\mathcal{C}$ is not viable. 
%Now, if $\mathcal{C}$ is viable, then 
\end{proof}

\section{Proof of Theorem \ref{theorem:simple-CPI}}
\label{appendix:CPI}

This appendix presents the proof of Theorem \ref{theorem:simple-CPI}. As a preliminary, we present background on triangulations of semi-algebraic sets. More details can be found in \cite[Ch. 9.2]{bochnak2013real}. 

Let $a_{0},\ldots,a_{k}$ be $(k+1)$ affinely independent points in $\mathbb{R}^{n}$. The $k$-simplex $[a_{0},\ldots,a_{k}]$ is the set of $x \in \mathbb{R}^{n}$ such that there exist nonnegative real numbers $\lambda_{0},\ldots,\lambda_{k}$ with $$\sum_{i=0}^{k}{\lambda_{i}} = 1, \quad x = \sum_{i=0}^{k}{\lambda_{i}a_{i}}.$$ For any subset $a_{i_{0}},\ldots,a_{i_{l}}$ of $\{a_{0},\ldots,a_{k}\}$, the $l$-simplex $[a_{i_{0}},\ldots,a_{i_{l}}]$ is a \emph{face} of $[a_{0},\ldots,a_{k}]$. If $\sigma$ is a simplex, then we denote by $\sigma^{0}$ the set of points in $\sigma$ whose coefficients $\lambda_{i}$ are all positive, and denote $\sigma^{0}$ as an open simplex.

\begin{definition}
\label{def:simp-comp}
A \emph{simplicial complex} $K = \{\sigma_{i} : i=1,\ldots,q\}$ is a collection of simplices such that, for each $\sigma_{i}$, all faces of $\sigma_{i}$ are included in $K$ and for all $i,j=1,\ldots,p$, either $\sigma_{i} \cap \sigma_{j} = \emptyset$ or $\sigma_{i} \cap \sigma_{j}$ is a common face of $\sigma_{i}$ and $\sigma_{j}$.
\end{definition}

The realization of $K$ is denoted $|K|$ and defined by $$|K| = \bigcup_{i=1}^{q}{\sigma_{i}}.$$ The open simplices $\sigma_{i}^{0}$ form a partition of $K$. The following theorem defines triangulations of semialgebraic sets.

\begin{theorem}[\cite{bochnak2013real}, Theorem 9.2.1]
    \label{theorem:semi-algebraic-SA}
    Every closed and bounded semi-algebraic set $\mathcal{C} \subseteq \mathbb{R}^{n}$ is semi-algebraically triangulable, i.e., there exists a finite simplicial complex $K = \{\sigma_{i} : i=1,\ldots,q\}$ and a semi-algebraic homeomorphism $\Phi : |K| \rightarrow \mathcal{C}$. Moreover, given a finite family $S_{1},\ldots,S_{r}$ of semi-algebraic subsets of $\mathcal{C}$, we can select a simplicial complex $K$ and a semi-algebraic triangulation $\Phi: |K| \rightarrow \mathcal{C}$ such that every $S_{j}$ is the union of a set of open simplicies $\Phi(\sigma_{i}^{0})$.
\end{theorem}

We are now ready to provide the proof of Theorem \ref{theorem:simple-CPI}.

\begin{proof}
The approach of the proof is to construct a feedback controller $\mu: \mathcal{C} \rightarrow \mathcal{U}$ that renders $\mathcal{C}$ positive invariant. First, for each $T \subseteq \{1,\ldots,r\}$, define $$\mathcal{C}_{T} = \mathcal{C} \cap \{x: b_{i}(x) = 0 \ \forall i \in T\}.$$ For each $x$, define $T$ to be the maximal set $T$ such that $x \in \mathcal{C}_{T}$, and let $u(x)$ satisfy $u(x) \in \mathcal{U}$ and $\frac{\partial b_{i}}{\partial x}(f(x) + g(x)u(x)) > 0$ for all $i \in T$. By continuity, there exists a ball $B(x)$ centered at $x$ such that, for all $z \in B(x)$, $\frac{\partial b_{i}}{\partial x}(f(z) + g(z)u(x)) > 0$. Since $\mathcal{C}_{T}$ is compact, there exists a finite set of balls $B(x_{T,1}),\ldots,B(x_{T,M_{T}})$ such that $\mathcal{C}_{T}$ is contained in their union. Furthermore, there is a continuous and positive definite function $d_{T}: \mathcal{C}_{T} \rightarrow \mathbb{R}$ such that, for any $x \in \mathcal{C}_{T}$, the ball centered at $x$ with radius $d_{T}$ is contained in $B(x_{T,i})$ for some $i \in 1,\ldots,M_{T}$. Let $\overline{d} = \min_{T}{\min{\{d_{T}(x) : x \in \mathcal{C}_{T}\}}}$, noting that $\overline{d} > 0$.

Let $Z_{1},\ldots,Z_{N}$ be a collection of balls centered at points in $\mathcal{C}$ with radius $\frac{\overline{d}}{2}$ such that $\mathcal{C} \subseteq \bigcup_{i=1}^{N}{Z_{i}}$. Consider the collection of semi-algebraic sets $\{\mathcal{C}_{T,i}: T \subseteq \{1,\ldots,r\}, i=1,\ldots,N\}$ defined by $\mathcal{C}_{T,i} = \mathcal{C}_{T} \cap K_{i}$. Let $(K,\Phi)$ be a semi-algebraic triangulation satisfying the conditions of Theorem \ref{theorem:semi-algebraic-SA}. We define the control policy $\mu$ as follows. First, suppose that $w_{jl}$ is a vertex of a simplex $\sigma_{j} \in K$, and let $T(w_{jl})$ be the maximal set $T$ with $\Phi(w_{jl}) \in \mathcal{C}_{T}$. We have that $\Phi(w_{jl}) \in B_{T}(x_{e})$ for some $e \in \{1,\ldots,M_{T(w_{jl})}\}$. We choose $\mu(\Phi(w_{jl})) = u(x_{e})$.

Now, suppose that $x \in \mathcal{C}_{T,i}$. By definition, we have $x \in \Phi(\sigma_{j}^{0})$ for an open simplex $\sigma_{j}^{0} = [w_{j1},\ldots,w_{jL}]$ such that $\Phi(\sigma_{j}^{0}) \subseteq \mathcal{C}_{T,i}$. Defining $a_{1}(x),\ldots,a_{L}(x)$ to be the barycentric coordinates of $x$, define $$\mu(x) = \sum_{k=1}^{L}{a_{k}(x)\mu(\Phi(w_{jk}))}.$$ Note that this function is continuous in $x$. We will show that $\mu(x)$ satisfies the conditions for positive invariance.

Let $w_{jl}$ be a vertex of $\sigma_{j}^{0}$. By definition of $K_{i}$, we have that $||x-\Phi(w_{jl})|| \leq \overline{d}$, and hence $x$ is contained in the ball centered at $\Phi(w_{jl})$ with radius $d(\Phi(w_{jl}))$. By definition of $d$, we then have that $x \in B_{T}(x_{e})$ where $x_{e}$ is defined as above, implying that $\frac{\partial b_{i}}{\partial x}(f(x) + g(x)\mu(\Phi(w_{jl})) > 0$. The control $\mu(x)$ is therefore a convex combination of values of $u$ that satisfy $\frac{\partial b_{i}}{\partial x}(f(x) + g(x)u) > 0$, implying that $\mu(x)$ satisfies the condition as well.

%The approach of the proof is to construct a feedback controller $\mu : \mathcal{C} \rightarrow \mathcal{U}$ that renders $\mathcal{C}$ positive invariant. First, for each $T \subseteq \{1,\ldots,r\}$, define the set $\mathcal{C}_{T}$ by $$\mathcal{C}_{T} = \mathcal{C} \cap \{x: b_{i}(x) = 0 \ \forall i \in T\} \cap \{x : b_{i}(x) > 0 \ \forall i \notin T\}.$$ Note that the $\mathcal{C}_{T}$ sets form a partition of $\partial \mathcal{C}$.  For each $x \in \mathcal{C}_{T}$, let $u(x)$ satisfy the conditions of the theorem.
 By continuity, there exists a ball $B(x)$ centered at $x$ such that, for all $z \in B(x)$, $\frac{\partial b_{i}}{\partial x}(f(z) + g(z)u(x)) > 0$ for all $i \in T$. By local compactness of $\mathcal{C}$, there is a countable set of balls $B(x)$, indexed $B_{T,1},B_{T,2},\ldots,$ that covers $\mathcal{C}_T$ and such that each $x \in \mathcal{C}_T$ lies in at most finitely many balls. Moreover, for every $x \in \partial \mathcal{C}_T$, there exists $d(x) > 0$ and a ball $B_{T,k}$ such that $B(x,d(x)) \subseteq B_{T,k}$.

\end{proof}

\begin{proof}[Proof of Proposition \ref{prop:CPI-conditions}]
By Theorem \ref{theorem:simple-CPI}, it suffices to prove that, for each $x \in \partial \mathcal{C}$, there exists $u \in \mathcal{U}$ satisfying $\frac{\partial b_{i}}{\partial x}(f(x)+g(x)u) > 0$ for all $i$ with $b_{i}(x) = 0$. The existence of such a $u$ is equivalent to the conditions of the proposition by Motzkin's Transposition Theorem.
\end{proof}

\begin{IEEEbiography}[{\includegraphics[width=1in,height=1.25in,clip,keepaspectratio]{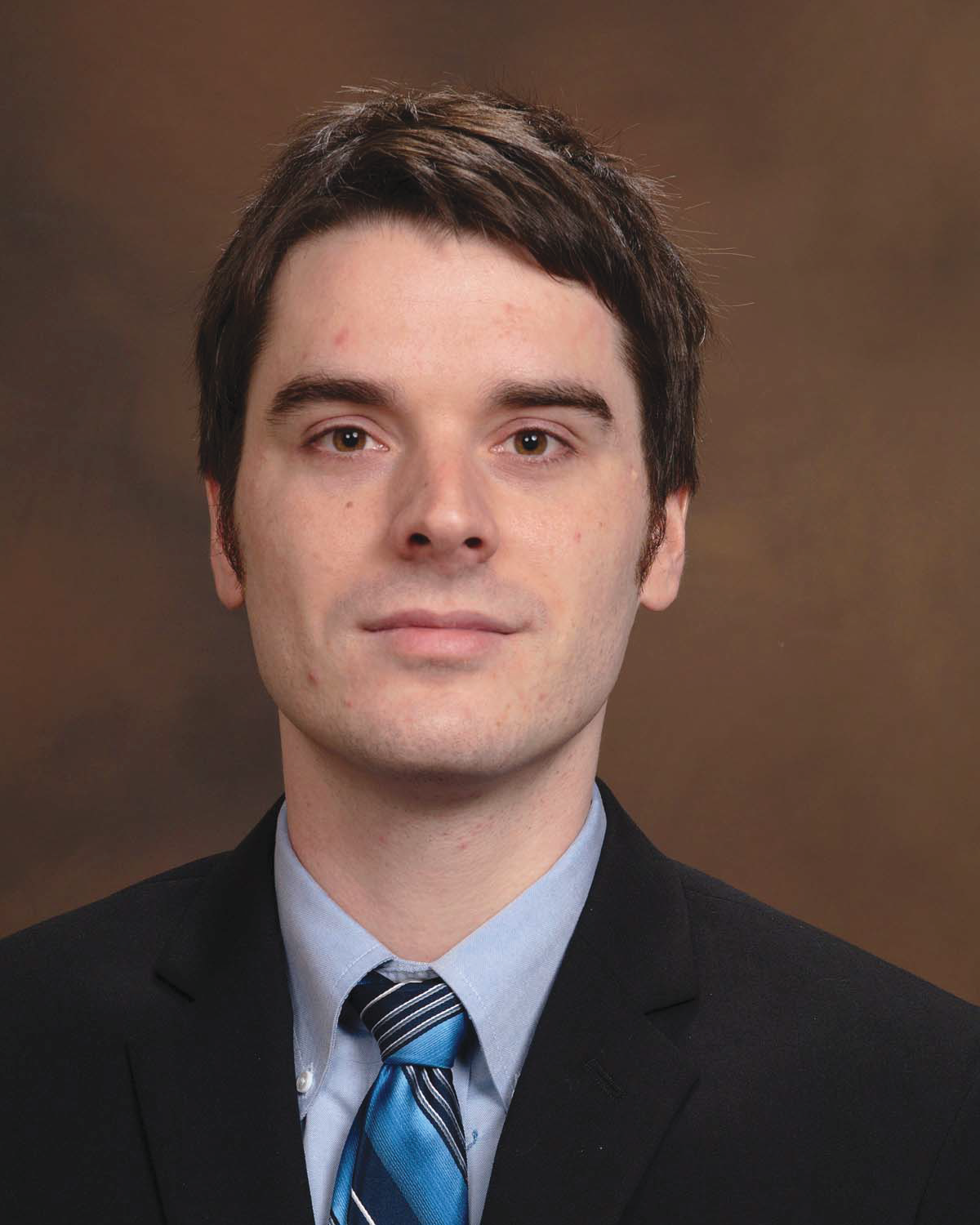}}]{Andrew Clark}
(SM’24) is an Associate Professor in the Department of Electrical and Systems Engineering at Washington University in St. Louis. He received the B.S.E. degree in Electrical Engineering and the M.S. degree in Mathematics from the University of Michigan - Ann Arbor in 2007 and 2008, respectively. He received the Ph.D. degree in Electrical Engineering from the Network Security Lab (NSL), Department of Electrical Engineering, at the University of Washington  - Seattle in 2014. He is author or co-author of the IEEE/IFIP William C. Carter award-winning paper
(2010), the WiOpt Best Paper (2012), the WiOpt Student Best Paper
(2014), and the GameSec Outstanding Paper (2018), and was a finalist for the IEEE CDC 2012 Best Student Paper Award and the ACM/ICCPS Best Paper Award (2016, 2018, 2020). He received an NSF CAREER award in 2020 and an AFOSR YIP award in 2022. His research interests include control and security of complex networks, safety of autonomous systems, submodular optimization, and control-theoretic modeling of network security threats.
\end{IEEEbiography}

\end{document}